\documentclass[11pt]{article}

\usepackage[margin=1in]{geometry}
\usepackage{amsmath,amssymb,amsthm,mathtools}
\usepackage{graphicx,tikz}
\usepackage{caption}
\usepackage{enumitem}
\usepackage[colorlinks=true,linkcolor=blue,citecolor=blue,urlcolor=blue]{hyperref}
\newtheorem{theorem}{Theorem}[section]
\newtheorem{lemma}[theorem]{Lemma}
\newtheorem{definition}[theorem]{Definition}
\newtheorem{proposition}[theorem]{Proposition}
\newtheorem{corollary}[theorem]{Corollary}
\newtheorem{fact}[theorem]{Fact}
\newtheorem{conjecture}[theorem]{Conjecture}
\theoremstyle{remark}
\newtheorem{remark}[theorem]{Remark}

\providecommand{\arcosh}{\operatorname{arcosh}}
\providecommand{\sinc}{\operatorname{sinc}}
\newcommand{\Tcal}{\mathcal T}

\newcommand{\Kcal}{\mathcal K}
\newcommand{\C}{\mathbb C}
\newcommand{\R}{\mathbb R}
\renewcommand{\Re}{\operatorname{Re}}

\title{Optimal Extrapolation Bounds for Sparse Fourier Sums}
\author{Ruizhe Zhang\\
\texttt{rzzhang@purdue.edu}\\
Purdue University}
\date{}

\begin{document}
\maketitle

\begin{abstract}
We prove a sharp extrapolation bound for $k$-sparse Fourier sums with arbitrary real frequencies and no separation assumption. For every $g(t)=\sum_{j=1}^k v_je^{i\lambda_jt}$ with $\lambda_j\in\mathbb R$ and every $x\geq 1$,
\begin{align*}
    |g(x)|\le k\exp(O(k\arcosh x))\|g\|_{L^2[-1,1]}.
\end{align*}
Near the endpoint at $x=1+\delta$, this gives an exponential factor of $\exp(O(k\sqrt\delta))$, improving on the $\exp(O(k^2\log k\cdot\delta))$ estimate of Chen and Price (ICALP 2019), and the exponent is optimal up to a constant factor.  Our theorem has algorithmic applications to recovering gapless Fourier-sparse signals: it improves the resolution of cluster-location algorithms and the sample complexity of Fourier interpolation algorithms.  In addition, it gives sharp bounds on the error transfer of sparse Fourier models just outside the training interval.
\end{abstract}

\section{Introduction}\label{sec:intro}
A signal with only $k$ significant spectral components admits a concise representation as a superposition of $k$ complex exponentials, and exploiting this structure is a central theme in sparse Fourier algorithms.  In the discrete setting, Fourier sparsity enables sublinear-time algorithms for identifying and estimating the dominant Fourier coefficients \cite{GGIMS2002,Iwen2010,HIKP2012SODA,HIKP2012STOC,IndykKapralov2014,Kapralov2016,NakosSongWang2019}.  In the continuous, off-grid setting, where the frequencies may be arbitrary real numbers, it underlies sparse Fourier recovery and interpolation algorithms \cite{PriceSong2015,CKPS2016,ChenPrice2019,SSWZ2023}; closely related super-resolution problems ask to recover sparse point sources from bandlimited measurements \cite{Donoho1992,CandesFernandezGranda2014,DemanetNguyen2015,Moitra2015,li2020super,li2022stability,ding2024esprit}.

A central difficulty in the continuous setting is the absence of a frequency gap.  When several frequencies are closer than the inverse length of the observation window, the corresponding exponentials become nearly linearly dependent: the associated Fourier--Vandermonde systems are ill-conditioned, and the individual frequencies may no longer be stably identifiable.  In the limit of a $k$-fold frequency collision, suitably rescaled finite differences of the exponentials converge to $1,t,\ldots,t^{k-1}$, so the clustered Fourier model approaches the space of polynomials of degree at most $k-1$.  In this regime, the natural goals are to reconstruct the signal with $L^2$-error guarantees \cite{CKPS2016,ChenPriceActive2019,SSWZ2023} and to estimate the location of the frequency cluster \cite{ChenPrice2019}.

This paper studies a deterministic extrapolation problem that underlies the gapless regime.  Let
$$
        g(t)=\sum_{j=1}^k v_j e^{i\lambda_jt},
        \qquad \lambda_j\in\mathbb R,
$$
be an arbitrary $k$-Fourier-sparse signal, and suppose that its energy on $[-1,1]$ is bounded. The extrapolation problem asks:
\begin{center}
\emph{How large can $g$ be at a point $x>1$ outside the observation interval $[-1,1]$?}
\end{center}

This problem was first studied by Chen and Price~\cite{ChenPrice2019}, who needed such an extrapolation bound to analyze their algorithm for locating a cluster of nearby frequencies. They proved that for $0\leq\delta\leq1$,
\begin{equation}\label{eq:cp19}
        |g(1+\delta)|
        \le
        \operatorname{poly}(k)\,
        e^{O(k^2\log k)\,\delta}\,
        \|g\|_{L^2[-1,1]} .
\end{equation}
However, the exponent in \eqref{eq:cp19} is unlikely to be sharp, as pointed out in~\cite[Remark 8.6]{ChenPrice2019}. As the frequencies collide, $k$-Fourier-sparse signals can approximate any polynomial of degree at most $k-1$. For polynomials, a similar extrapolation problem is well-studied: among all polynomials of degree at most $n$ that are bounded in absolute value by one on $[-1,1]$, the Chebyshev polynomial $T_n(x)=\cosh\bigl(n\arcosh(x)\bigr)$ is the largest at every point $x>1$~\cite{BorweinErdelyi1995}. In particular, $T_{k-1}(1+\delta)=\exp\bigl(\Theta(k\sqrt{\delta})\bigr)$ for $k^{-2}\leq \delta \leq 1$ ($k\geq 2$). Therefore, Chen and Price suggested that the optimal exponent in the extrapolation bound for $k$-Fourier-sparse signals is $k\sqrt{\delta}$ rather than $k^2\log(k)\delta$. 

The main contribution of this paper is to fully resolve this open problem by proving the following sharp extrapolation bound for Fourier signals with $k$ arbitrary real frequencies.

\begin{theorem}[Extrapolation bound for Fourier-sparse signals]\label{thm:main}
For every $k$-Fourier-sparse signal
$$
        g(t)=\sum_{j=1}^k v_j e^{i\lambda_jt},\qquad \lambda_j\in\mathbb R,
$$
and every $x\ge1$,
\begin{equation}\label{eq:main-global}
        |g(x)|\lesssim k\,e^{4k\arcosh (x)}\,\|g\|_{L^2[-1,1]}\,.
\end{equation}
Moreover, for $0\le\delta\le1$,
\begin{equation}\label{eq:main-edge}
        |g(1+\delta)|\lesssim k\,e^{4k\sqrt{\delta}}\,\|g\|_{L^2[-1,1]}\,.
\end{equation}
\end{theorem}
We also show that the exponent in~\eqref{eq:main-global} is optimal up to a constant factor, by converting Chebyshev polynomials to $k$-Fourier-sparse signals (Proposition~\ref{prop:sharpness}).

\subsection{Applications of the extrapolation inequality}
\label{subsec:intro-alg}
Theorem~\ref{thm:main} is not only an approximation theory result, but also has useful applications in theoretical computer science.  We apply it to three problems about learning a signal from limited observations: locating a cluster of nearby frequencies, reconstructing a signal from noisy samples, and predicting with a fitted model beyond the observation interval.  In each problem, the extrapolation bound controls the signal where it is not observed.

\subsubsection{Locating a cluster of nearby frequencies}
Consider the following basic question for the continuous sparse Fourier transform without a frequency gap.  Suppose that all frequencies of a $k$-Fourier-sparse signal $g$ lie in an interval $[f_0-\Delta,f_0+\Delta]$ of known half-width $\Delta$, whose center $f_0\in [-F,F]$ is unknown.  An algorithm may query noisy values $y(t)=g(t)+\zeta(t)$ at times $t\in[-T,T]$ of its choice, where the noise is guaranteed to have small energy but may be adversarial.  How accurately can the algorithm estimate $f_0$, and how many samples does it need?

Chen and Price~\cite{ChenPrice2019} obtained an estimation error of $\Delta+\widetilde O(k^3/T)$ with  roughly $\widetilde{O}\bigl(k\bigr)$ samples. Theorem~\ref{thm:intro-cluster} improves the resolution by a factor of $k$ with the same sample complexity.

\begin{theorem}[Cluster location; informal version of Corollary~\ref{cor:chen-price-main-improvement}]
\label{thm:intro-cluster}
There is an efficient algorithm that uses $O\bigl(k\log^2k\cdot\log(F/\Delta_*)\bigr)$ samples, where $\Delta_*:=\Delta+k^2/T$, and with constant probability returns an estimate $\widehat f_0$ with
\[
        |\widehat f_0-f_0|\leq O\Bigl(\Delta+\frac{k^2}{T}\Bigr).
\]
\end{theorem}

Our improvement comes from a \emph{filter function} that the algorithm applies to the signal: a function $H$, close to one on most of $[-T,T]$ and small outside it, by which the algorithm multiplies its observations.  The filter must do two things:
\begin{itemize}
    \item First, the filtered signal $gH$ should keep most of the energy of $g$ on $[-T,T]$ and have little energy outside $[-T,T]$, where no samples are taken.
    \item Second, the filter should not blur the cluster.  If $H$ has bandwidth $B$, that is, its Fourier transform is supported in $[-B,B]$, then each frequency of $gH$ lies within $B$ of a frequency of $g$.  So the cluster is broadened by at most $B$, which will enter the final error. 
\end{itemize}
The hard truncation $H(t)=\mathbf 1_{[-T,T]}(t)$ meets the first requirement perfectly but fails the second: it spreads every frequency over the whole line.  Conversely, a filter of finite bandwidth cannot vanish outside $[-T,T]$, so its tails must decay faster than the unobserved signal can grow.  This is where the extrapolation bound enters.  With $T=1$, Theorem~\ref{thm:main} shows that the signal grows at most like $e^{O(k\sqrt r)}$ at distance $r\leq1$ beyond the endpoint, while a filter of bandwidth $B$ can decay like $e^{-\Omega(\sqrt{Br})}$ (Section~\ref{subsec:chen-price-app}), so $B=\Theta(k^2)$ suffices.  No filter can do better: every filter that meets both requirements has $B=\Omega(k^2/T)$ (Proposition~\ref{prop:bandwidth-floor}).

\subsubsection{Reconstructing a signal from noisy samples}\label{sec:interpolation_corollary}
\emph{Fourier interpolation} aims to reconstruct the entire signal without any assumptions about the gaps between frequencies.  The $k$ frequencies may be arbitrary points of a known band $[-F,F]$.  Given noisy samples at chosen times in $[-T,T]$ with a signal-to-noise ratio $1/\gamma$, the algorithm must output a function $\widehat g$, not necessarily $k$-Fourier-sparse, that is close to $g$ in $L^2[-T,T]$.

The first polynomial-time algorithm for this problem was given by Chen, Kane, Price, and Song~\cite{CKPS2016}, and Song, Sun, Weinstein, and Zhang~\cite{SSWZ2023} reduced the sample complexity to $\widetilde O(k^4)$.  Very recently, Cai, Chen, Shao, and Wang~\cite{CaiChenShaoWang2026} improved this to $\widetilde O(k^{15/4})$, and to $\widetilde O(k^3)$ assuming a bound on the growth of $g$ outside $[-T,T]$ that they conjectured~\cite[Conjecture~1.3]{CaiChenShaoWang2026}.  Theorem~\ref{thm:main} proves this conjecture, so their algorithm achieves the following guarantee unconditionally.
\begin{corollary}[Cubic-sample Fourier interpolation~\cite{CaiChenShaoWang2026}; informal version of Corollary~\ref{cor:CCSW-interpolation}]
\label{thm:intro-interpolation}
There is an algorithm that uses $\widetilde O(k^3)$ samples and $\widetilde O(k^{3\omega})$ time and, with constant probability, outputs $\widehat g$ satisfying
\[
        \|\widehat g-g\|_{L^2[-T,T]}^2\leq O(\gamma)\,\|g\|_{L^2[-T,T]}^2 .
\]
Here, $\omega$ is the fast matrix-multiplication exponent, and $\widetilde O$ hides factors polylogarithmic in $kFT$.
\end{corollary}

Their algorithm first finds $O(k)$ candidate frequencies, representing the groups of nearby frequencies that carry almost all of the signal's energy.  It then describes the signal near each candidate by an exponential times a polynomial, and the less accurately the candidates are localized, the higher the degree of these polynomials must be.  As in the first application, a sharper growth bound gives better localization, and our extrapolation bound in Theorem~\ref{thm:main} reduces the localization error from $\widetilde O(k^{11/4}/T)$ to $\widetilde O(k^2/T)$.  This lowers the polynomial degrees to $\widetilde O(k^2)$, thereby reducing the number of coefficients to fit and the number of samples from $\widetilde O(k^{15/4})$ to $\widetilde O(k^3)$.  See Section~\ref{subsec:fourier-interpolation-app} for more details.

In subsequent work, Cai, Chen, and Shao~\cite{CaiChenShao2026} use our extrapolation bound to obtain a further improved Fourier interpolation algorithm with $\widetilde O(k^2)$ samples and $\widetilde O(k^5)$ time.  Their key step is to approximate every $k$-Fourier-sparse signal by a signal whose Fourier spectrum is supported on a union of intervals of total length $\widetilde O(k^2)$.  Our bound is used in this step to control the part of the signal outside the observation interval after it is multiplied by a narrow-band filter. 

\subsubsection{Predicting beyond the observation interval}
Our third application concerns out-of-domain generalization, that is, prediction outside the region supported by the training data.  Suppose that a model $\widehat g$ with at most $k$ frequencies has been fitted to noisy observations of a $k$-sparse signal $g_\star$ on $[-1,1]$, and that its average squared error on this training interval is
\[
        \epsilon_{\mathrm{train}}^2:=\frac12\int_{-1}^1|\widehat g(t)-g_\star(t)|^2\,dt .
\]
How accurate is $\widehat g$ just outside $[-1,1]$?  For general model classes, a small error on $[-1,1]$ says nothing about the error outside it.  For sparse Fourier models, however, the error $\widehat g-g_\star$ is itself a sum of at most $2k$ exponentials, so Theorem~\ref{thm:main} applies to it.

\begin{theorem}[Exterior prediction; informal version of Corollary~\ref{cor:active-exterior-transfer} and Proposition~\ref{prop:exterior-leverage}]
\label{thm:intro-prediction}
For every $0\leq\delta\leq1$, the squared prediction error on $\mathcal X_\delta:=[-1-\delta,-1]\cup[1,1+\delta]$ satisfies
\[
        \int_{\mathcal X_\delta}|\widehat g(x)-g_\star(x)|^2\,dx
        \leq(1+k\sqrt\delta)\,e^{O(k\sqrt\delta)}\,\epsilon_{\mathrm{train}}^2 .
\]
For $k^{-2}\leq\delta\leq1$, the factor $e^{O(k\sqrt\delta)}$ cannot be improved beyond constants in the exponent and polynomial factors in $k$.
\end{theorem}

The guarantee holds for any fitting procedure that returns a sparse Fourier model, such as harmonic regression, random Fourier features, and shallow neural networks with sine activations.  The error amplification is polynomial in $k$ as long as $\delta=O((\log k/k)^2)$, and it is $e^{\Theta(k)}$ at $\delta=\Theta(1)$.  One can also ask a pointwise question: how large can the squared error at a single point $x$ outside $[-1,1]$ be, relative to the mean squared error $\epsilon_{\mathrm{train}}^2$ on $[-1,1]$?  For a fixed set of frequencies, the worst case of this ratio over all error functions with those frequencies is the statistical \emph{leverage score} of $x$, the quantity that guides sampling in active regression.  Theorem~\ref{thm:main} therefore bounds leverage scores outside the training interval, uniformly over all frequency sets.  Section~\ref{subsec:active-learning-app} studies both this pointwise quantity and the transfer of integrated error. 
\subsection{Technical overview}
\label{subsec:technical-overview}
Define the set of $k$-Fourier-sparse signals
\begin{align*}
    {\cal T}_k=\left\{\sum_{j=1}^k v_j e^{i\lambda_j t}:\lambda_j\in \R,v_j\in \C~~\forall j\in [k]\right\}.
\end{align*}
For any $\delta>0$ and $g\in {\cal T}_k$, our goal is to prove the following extrapolation bound
\begin{align}\label{eq:tech_overview_goal}
    |g(1+\delta)|\leq ke^{O(k\arcosh (1+\delta))}\,\|g\|_{L^2[-1,1]}.
\end{align}
We distinguish three regimes of $\delta$.
\begin{enumerate}
    \item When $\delta<k^{-2}$, the target exponent $k\arcosh(1+\delta)=\Theta(k\sqrt\delta)$ in~\eqref{eq:tech_overview_goal} is $O(1)$, and a two-line argument that slides the observation window outward suffices (Section~\ref{ov:naive}).
    \item When $\delta>1$, the target exponent $k\arcosh(1+\delta)$ is $\Omega(k)$, so losing a factor of $e^{O(k)}$ is harmless, and a classical Tur\'an-type comparison suffices (Section~\ref{ov:outer}).
    \item When $k^{-2}\leq\delta\leq1$, the target exponent $\Theta(k\sqrt\delta)$ ranges between $\Theta(1)$ and $\Theta(k)$. This is the difficult regime, where both arguments above lose too much: sliding the window gives only $e^{O(k^2\delta)}$, and the Tur\'an-type comparison gives only $e^{O(k)}$.
\end{enumerate}

To handle the third regime of $\delta$, since~\eqref{eq:tech_overview_goal} holds trivially for a $1$-Fourier-sparse signal, we aim to reduce a $k$-Fourier-sparse signal to a $(k-1)$-Fourier-sparse signal while controlling its energy. Note that na\"ively removing one Fourier component cannot work: if
\begin{align*}
    \qquad g(t)=\sum_{j=1}^k v_j e^{i\lambda_j t}\qquad \text{and}\qquad h(t)=g(t)-v_ke^{i\lambda_kt}=\sum_{j=1}^{k-1} v_j e^{i\lambda_j t},
\end{align*}
then the energy difference $\left|\|h\|_{L^2[-1,1]}-\|g\|_{L^2[-1,1]}\right|$ can be arbitrarily large due to the removal of cancellations (or ``interference'') between close frequencies. We instead construct the following frequency elimination process:
\begin{align}\label{eq:tech_overview_freq_eliminate}
    b(t):=\alpha\int_0^\infty e^{-\alpha s}e^{i\lambda s}g(t-s)\,ds,
        \qquad h(t):=g(t)-b(t),
\end{align}
where $\alpha>0$ is a parameter and $\lambda$ is the removed frequency of $g$. Moreover, the original signal $g$ can be recovered from $h$ via the following differential equation with the boundary value $b(1)$:
\begin{align}\label{eq:tech_overview_reconstruct_ode}
    b'=i\lambda b+\alpha h,
\end{align}
which can be solved explicitly:
\begin{align}\label{eq:tech_overview_reconstruct}
g(1+\delta)=h(1+\delta)+b(1)e^{i\lambda \delta} +\alpha\int_0^\delta e^{i\lambda(\delta-s)}h(1+s)\,ds\qquad \forall \delta>0.
\end{align}
Then, we define a notion of \emph{$\alpha$-weighted energy}
\begin{align}\label{eq:tech_overview_def_energy}
    {\cal E}_\alpha(g):=\int_{-\infty}^1 e^{\alpha(t-1)}|g(t)|^2\,dt,
\end{align}
and establish the following identity for the energy decrease of one frequency elimination:
\begin{align}\label{eq:tech_overview_energy_diff}
    {\cal E}_\alpha(g)-{\cal E}_\alpha(h)=\frac{|b(1)|^2}{\alpha}.
\end{align}
See Section~\ref{ov:ema} for a more detailed discussion about this frequency elimination process.

Now, we can repeatedly apply frequency eliminations to obtain a sequence of signals $f_0=g, f_1, f_2, \dots, f_k=0$, where $f_j$ has Fourier sparsity $k-j$. Let $b_j:=f_j - f_{j+1}$ for $0\leq j<k$. By telescoping summation of~\eqref{eq:tech_overview_energy_diff}, we have
\begin{align}\label{eq:tech_overview_energy_sum}
    {\cal E}_\alpha(g) = \frac{1}{\alpha}\left(|b_0(1)|^2+\cdots +|b_{k-1}(1)|^2\right).
\end{align}
Note that $\alpha{\cal E}_\alpha(g)$ can be recognized as the expectation of $|g(t)|^2$ under an exponential distribution over $(-\infty, 1]$. To relate it to $\|g\|_{L^2[-1,1]}$, we prove the following concentration result:
\begin{align*}
    \int_{1-8k/\alpha}^{1} e^{\alpha(t-1)}|g(t)|^2\,dt\geq \frac{1}{2} {\cal E}_\alpha(g).
\end{align*}
In particular, for $\alpha \geq 4k$, we have
\begin{align}\label{eq:tech_overview_energy_convert}
    {\cal E}_\alpha(g)\leq 2\int_{1-8k/\alpha}^{1} e^{\alpha(t-1)}|g(t)|^2\,dt\leq 2\int_{-1}^{1} e^{\alpha(t-1)}|g(t)|^2\,dt\leq 2\|g\|_{L^2[-1,1]}^2.
\end{align}
On the other hand, by unfolding the reconstruction formula~\eqref{eq:tech_overview_reconstruct}, we have
\begin{align}
    |f_{k-1}(1+\delta)|=&~|b_{k-1}(1)|,\notag\\
    |f_{k-2}(1+\delta)|\leq&~ |b_{k-2}(1)|+|b_{k-1}(1)|(1+\alpha \delta),\notag\\
    |f_{k-3}(1+\delta)| \leq &~ |b_{k-3}(1)|+|b_{k-2}(1)|(1+\alpha \delta) + |b_{k-1}(1)|(1+2\alpha\delta +\alpha^2\delta^2/2),\notag\\
    \vdots&\notag\\
    |g(1+\delta)| \leq&~ \sum_{j=0}^{k-1} |b_j(1)|\Phi_j(\delta),\qquad \Phi_j(\delta):=\sum_{\ell=0}^j \binom{j}{\ell}\frac{(\alpha \delta)^\ell}{\ell!}\leq e^{2\sqrt{\alpha \delta j}}.\label{eq:tech_overview_g_bound}
\end{align}

Finally, we complete the proof of~\eqref{eq:tech_overview_goal} by combining~\eqref{eq:tech_overview_energy_sum}--\eqref{eq:tech_overview_g_bound} with $\alpha=4k$, and using Cauchy--Schwarz:
\begin{align*}
    |g(1+\delta)|\leq \left(\sum_{j=0}^{k-1}|b_j(1)|^2\right)^{1/2}\left(\sum_{j=0}^{k-1}\Phi_j(\delta)^2\right)^{1/2}\leq
        \sqrt{2\alpha k}\,
        e^{2\sqrt{k\alpha\delta}}
        \|g\|_{L^2[-1,1]}\lesssim ke^{4k\sqrt{\delta}}\|g\|_{L^2[-1,1]}.
\end{align*}

\subsubsection{The sliding window argument}\label{ov:naive}
Chen and Price~\cite{ChenPrice2019} proved their extrapolation bound~\eqref{eq:cp19} using a linear recurrence:
\begin{align}\label{eq:tech_overview_cp19_recurrence}
    g(t)=\sum_{j=1}^{d}c_j\,g(t-2j/d)\qquad\forall t\in\R, \qquad |c_j|\leq C, \qquad d=O(k^2\log k).
\end{align}
For $0<r\leq2/d$, all $d$ values on the right-hand side at $t=1+r$ lie in the observation window $[-1,1]$, so $|g(1+r)|\leq\operatorname{poly}(k)\,\|g\|_{L^2[-1,1]}$.  Sliding the window to the right by $2/d$ and applying the recurrence again extends the bound to $2/d<r\leq4/d$. They showed that each such step multiplies the bound by only a constant factor.  Reaching $r=\delta$ takes $d\delta/2$ steps, which gives the extrapolation bound 
\begin{align*}
    |g(1+\delta)|\leq\operatorname{poly}(k)\,e^{O(d\delta)}\|g\|_{L^2[-1,1]}=\operatorname{poly}(k)\,e^{O(k^2\log(k)\delta)}\|g\|_{L^2[-1,1]}.
\end{align*}

This argument can be sharpened slightly with the Nikolskii inequality (Lemma~\ref{lem:nikolskii}): for every $g\in\Tcal_k$ and every $b\in\R$,
\begin{equation}\label{eq:ov-nikolskii}
         \max_{t\in[b-2,b]}|g(t)|^2\leq\frac{\pi^2k^2}{8}\int_{b-2}^{b}|g(t)|^2\,dt.
\end{equation}
Let $E(s):=\int_{s-1}^{s+1}|g(t)|^2\,dt$ be the energy in the window $[s-1,s+1]$, so that $E(0)=\|g\|_{L^2[-1,1]}^2$.  
By~\eqref{eq:ov-nikolskii} with $b=s+1$, we have 
\[
        E'(s)=|g(s+1)|^2-|g(s-1)|^2\leq|g(s+1)|^2\leq\frac{\pi^2k^2}{8}\,E(s),
\]
which implies that
\begin{align*}
    E(\delta)\leq e^{O(k^2\delta)}E(0)=e^{O(k^2\delta)}\|g\|_{L^2[-1,1]}^2.
\end{align*}
Applying~\eqref{eq:ov-nikolskii} once more with $b=1+\delta$ gives
\begin{equation}\label{eq:elementary-bound}
        |g(1+\delta)|\leq O(k)\cdot \sqrt{E(\delta)}\leq k e^{O(k^2\delta)}\,\|g\|_{L^2[-1,1]} .
\end{equation}
This shaves the $\log k$ factor in the exponent in~\eqref{eq:cp19}, and it is sharp for $\delta\lesssim k^{-2}$.  However, for larger $\delta$, this argument is too loose.  \eqref{eq:ov-nikolskii} gives the worst-case ratio $|g(s+1)|^2/E(s)\leq O(k^2)$ at a single point $s$, while the argument uses it at every $s\in[0,\delta]$.  A $k$-Fourier-sparse signal cannot attain it everywhere.  Our proof does not need this argument since the method we develop for $k^{-2}\leq\delta\leq1$ can also cover $\delta<k^{-2}$.

\subsubsection{The Tur\'an-type comparison}\label{ov:outer}
When $\delta>1$, the point $x=1+\delta$ is far from the observation window, and we compare $g(x)$ directly with values of $g$ inside $[-1,1]$, without passing through intermediate points.  The tool is Tur\'an's lemma for sums of geometric sequences (Lemma~\ref{lem:turan-finite}; see~\cite{Turan1984,Montgomery1994}): if $s_n=\sum_{j=1}^kb_jz_j^n$ with $|z_j|\geq1$ for all $j$, then for every $m\geq0$,
\begin{equation}\label{eq:tech_overview_turan}
    |s_0|\leq k\,2^{k-1}\binom{m+k-1}{k-1}\max_{m<n\leq m+k}|s_n| .
\end{equation}
In words, a sum of $k$ non-decaying geometric sequences cannot be large at index $0$ while it is small on $k$ consecutive later indices.

To apply this, sample $g$ on the grid $t_n:=x-nh$ with spacing $h=2/(k+1)$.  Since $e^{i\lambda_jt_n}=e^{i\lambda_jx}(e^{-i\lambda_jh})^n$, the samples $s_n:=g(t_n)$ form a sum of $k$ geometric sequences whose ratios $e^{-i\lambda_jh}$ have modulus one, and $s_0=g(x)$.  The $k$ consecutive points $t_{m+1},\ldots,t_{m+k}$ lie in $[-1,1]$ for some $m\leq(k+1)(x-1)/2$, and then the binomial coefficient is at most $(Cx)^{k-1}$.  Hence \eqref{eq:tech_overview_turan} and the Nikolskii inequality \eqref{eq:ov-nikolskii} give
\begin{equation*}
    |g(x)|\leq k\,(Cx)^{k}\,\|g\|_{L^\infty[-1,1]}\leq k^2(Cx)^{k}\,\|g\|_{L^2[-1,1]}
\end{equation*}
for an absolute constant $C$ (Proposition~\ref{prop:outer}). Note that when $x\geq2$, the factor $(Cx)^k=e^{O(k\log x)}=e^{O(k\arcosh x)}$ has the target exponent in~\eqref{eq:tech_overview_goal}. However, when $x$ is close to $1$, the $O(k)$ exponent in this bound does not improve, while the target exponent is $O(k\sqrt\delta)$.

One might hope to avoid this loss by comparing $g(x)$ with all values of $g$ on $[-1,1]$ rather than with $k$ samples.  The continuous version of Tur\'an's lemma due to Nazarov~\cite{Nazarov1994} (see also~\cite{BrudnyiYomdin2016}) does exactly this: for every interval $I$, every measurable set $E\subseteq I$ of positive length, and every $g\in\Tcal_k$,
\begin{equation}\label{eq:tech_overview_nazarov}
    \sup_{t\in I}|g(t)|\leq O(|I|/|E|)^{k-1}\,\sup_{t\in E}|g(t)|.
\end{equation}
Taking $I=[-1,1+\delta]$ and $E=[-1,1]$ gives
\[
    |g(1+\delta)|\leq O(1+\delta)^{k-1}\,\|g\|_{L^\infty[-1,1]}.
\]
As $\delta\to0$, the factor is still $e^{O(k)}$, so the loss near the endpoint is not an artifact of sampling.

\subsubsection{Frequency elimination}\label{ov:ema}
We briefly discuss the intuition behind the frequency elimination process~\eqref{eq:tech_overview_freq_eliminate}. Our goal is to remove one frequency while controlling the energy, and we have discussed that directly subtracting one Fourier component cannot work. 

Another approach is differentiation: for any $g\in {\cal T}_k$, 
\begin{align*}
    D_\lambda g:=g'-i\lambda g.
\end{align*}
Then, each Fourier component with frequency $\mu$ is multiplied by $i\nu=i(\mu-\lambda)$. Although $D_\lambda$ does not depend on the coefficients, the multiplier $i\nu$ is unbounded. An ideal multiplier should vanish at $\nu=0$ to eliminate the selected frequency, and remain bounded when $|\nu|$ is large to control the energy. A natural choice is the rational function
\begin{align*}
    r_\alpha(\nu):=\frac{i\nu}{\alpha + i\nu},\qquad \alpha >0.
\end{align*}
We have
\begin{align*}
    r_\alpha(0)=0,\qquad |r_\alpha(\nu)|=\frac{|\nu|}{\sqrt{\alpha^2 + |\nu|^2}}\rightarrow 1\quad \text{as}~~|\nu|\rightarrow\infty. 
\end{align*}
This motivates us to define the residual signal
\begin{align*}
   g(t)=\sum_{j=1}^k v_j e^{i\lambda_j t}~~\longmapsto~~  h(t)=\sum_{j=1}^{k} v_j r_\alpha(\lambda_j-\lambda)e^{i\lambda_j t}=\sum_{j=1}^{k} v_j \frac{i(\lambda_j-\lambda)}{\alpha + i(\lambda_j-\lambda)}e^{i\lambda_j t}.
\end{align*}
Then, we have
\begin{align*}
    b(t)=g(t)-h(t)=\sum_{j=1}^{k} v_j \frac{\alpha}{\alpha + i(\lambda_j-\lambda)}e^{i\lambda_j t}.
\end{align*}
If we apply $D_\lambda$ to $b$, we obtain the differential equation for reconstruction~\eqref{eq:tech_overview_reconstruct_ode}:
\begin{align*}
    b'-i\lambda b=\sum_{j=1}^{k} v_j \frac{\alpha i(\lambda_j-\lambda)}{\alpha + i(\lambda_j-\lambda)}e^{i\lambda_j t}=\alpha h.
\end{align*}
The integral form of $b$ in~\eqref{eq:tech_overview_freq_eliminate} can be derived using the following identity:
\begin{align*}
    \frac{\alpha}{\alpha + i\nu}=\alpha\int_0^\infty e^{-\alpha s}e^{-i\nu s}\,ds\qquad \forall \nu\in \R,
\end{align*}
which implies that
\begin{align*}
    b(t)=\sum_{j=1}^{k} v_j\left(\alpha\int_0^\infty e^{-\alpha s}e^{-i(\lambda_j-\lambda) s}\,ds\right) e^{i\lambda_j t}=\alpha\int_0^\infty e^{-\alpha s}e^{i\lambda s}g(t-s)\,ds.
\end{align*}

It remains to control $b(1)$ (the boundary condition of the differential equation) and the energy of the residual signal $h(t)$. Using $b'=i\lambda b + \alpha h$ and $g=h+b$, we have
\begin{align*}
    \left(|b|^2\right)'=2\alpha \Re(h\overline{b})=\alpha\left(|g|^2-|h|^2-|b|^2\right),
\end{align*}
where $|g|^2-|h|^2$ is the energy difference we would like to integrate. To absorb the $-|b|^2$ term, we introduce a weight function $w$:
\begin{align*}
    \left(w|b|^2\right)'=\alpha w\left(|g|^2-|h|^2\right)+\left(w'-\alpha w\right)|b|^2.
\end{align*}
Thus, if we require $w'=\alpha w$ and a normalization condition $w(1)=1$, we have
\begin{align*}
    w_\alpha(t)=e^{\alpha(t-1)},\qquad \left(w_\alpha|b|^2\right)'=\alpha w_\alpha\left(|g|^2-|h|^2\right).
\end{align*}
Integrating from $-\infty$ to $1$ gives the weighted energy decrease identity~\eqref{eq:tech_overview_energy_diff}:
\begin{align*}
    \underbrace{\int_{-\infty}^1 e^{\alpha (t-1)} |g(t)|^2\,dt}_{{\cal E}_\alpha(g)}-\underbrace{\int_{-\infty}^1 e^{\alpha (t-1)} |h(t)|^2\,dt}_{{\cal E}_\alpha(h)}=\alpha^{-1}\int_{-\infty}^1 \left(e^{\alpha(t-1)}|b(t)|^2\right)'\,dt=\frac{|b(1)|^2}{\alpha}.
\end{align*}

\subsubsection{The Takenaka--Malmquist basis and the extremal signal}\label{ov:collision}
Notice that the identity~\eqref{eq:tech_overview_energy_sum} has the form of Parseval's identity: the $\alpha$-weighted energy of $g$ is the sum of the squares of the $k$ numbers $b_m(1)/\sqrt\alpha$.  This is no coincidence.  Equip the span of the exponentials with the weighted inner product
\[
    \langle f_1,f_2\rangle_\alpha:=\int_{-\infty}^1 f_1(t)\,\overline{f_2(t)}\,e^{\alpha(t-1)}\,dt,
\]
so that ${\cal E}_\alpha(f)=\langle f,f\rangle_\alpha$, and write one reconstruction step as an operator,
\[
    ({\cal R}_\lambda h)(t):=h(t)+\alpha\int_1^t e^{i\lambda(t-s)}h(s)\,ds.
\]
The reconstruction formula~\eqref{eq:tech_overview_reconstruct} can be written as $f_j=b_j(1)\,e^{i\lambda_{j+1}(t-1)}+{\cal R}_{\lambda_{j+1}}f_{j+1}$.  Unwinding it from $f_0=g$ to $f_k=0$ gives
\[
    g=\sum_{m=1}^{k}\frac{b_{m-1}(1)}{\sqrt\alpha}\,q_m,
    \qquad
    q_m:=\sqrt\alpha\,{\cal R}_{\lambda_1}{\cal R}_{\lambda_2}\cdots {\cal R}_{\lambda_{m-1}}\bigl(e^{i\lambda_m(t-1)}\bigr).
\]
The functions $q_1,\ldots,q_{k}$ form an orthonormal basis of the span of $e^{i\lambda_1t},\ldots,e^{i\lambda_{k}t}$ for $\langle\cdot,\cdot\rangle_\alpha$.  This is the classical Takenaka--Malmquist basis~\cite{Takenaka1925,Malmquist1926,Nikolski2002,GarciaRoss2015}, which is usually constructed in a Hardy space.  Our proof does not need this interpretation; see Appendix~\ref{app:analytic} for more details about this connection.

Using this set of basis functions, we can construct the extremal signal of our weighted bound $|g(1+\delta)|^2\leq\alpha\,{\cal E}_\alpha(g)\sum_{m<k}\Phi_m(\delta)^2$. It uses two inequalities: the triangle inequality in~\eqref{eq:tech_overview_g_bound} and Cauchy--Schwarz.  Both become equalities in the same limit as all frequencies approach a common value $\lambda$.  The phases $e^{i\lambda_j(\delta-s)}$ in~\eqref{eq:tech_overview_reconstruct} then agree, so no cancellation occurs, and the basis functions become
\[
    q_m(t)=\sqrt\alpha\,e^{i\lambda(t-1)}L_{m-1}\bigl(\alpha(1-t)\bigr),
\]
where $L_{m-1}$ is the ordinary Laguerre polynomial (Definition~\ref{def:def_laguerre}).   Since $L_{m-1}(-\alpha\delta)=\Phi_{m-1}(\delta)$, each basis function then satisfies $|q_m(1+\delta)|=\sqrt\alpha\,\Phi_{m-1}(\delta)$.  Cauchy--Schwarz is an equality when the $m$-th coordinate is proportional to $\Phi_{m-1}(\delta)$, for $1\leq m\leq k$, which gives the extremal signal
\[
    f_\star(t)=e^{i\lambda(t-1)}\,p_\star\bigl(\alpha(1-t)\bigr),
    \qquad
    p_\star(y):=\sum_{m=0}^{k-1}L_m(-\alpha\delta)\,L_m(y).
\]
Hence, the weighted bound is exact:
\begin{equation}\label{eq:overview-exact}
    \sup_{0\neq g\in{\cal T}_k}\frac{|g(1+\delta)|^2}{{\cal E}_\alpha(g)}
    =\alpha\sum_{m=0}^{k-1}\Phi_m(\delta)^2,
\end{equation}
with the supremum approached as all frequencies collide.  

For the observed norm on $[-1,1]$, the sharp comparison remains open in general. If $P_m$ denotes the Legendre polynomial normalized by $P_m(1)=1$, the frequency collision limit gives
\[
        \sup_{0\neq g\in\Tcal_k}
        \frac{|g(x)|^2}{\|g\|_{L^2[-1,1]}^2}
        \geq
        \sum_{m=0}^{k-1}\frac{2m+1}{2}P_m(x)^2,
        \qquad x\geq1.
\]
A natural question is whether equality holds. This would identify the exact finite-interval extrapolation constant and show that it is approached by coalescing frequencies.

\subsection*{Organization.}
Section~\ref{sec:prelim} collects the notation and the classical results we use.  Section~\ref{sec:upper} provides the proof of Theorem~\ref{thm:main}.  Section~\ref{sec:sharpness} shows the optimality of the exponent in our extrapolation bound. Section~\ref{sec:applications} presents the three applications: clustered-frequency estimation, Fourier interpolation, and prediction beyond the observation interval.  Appendix~\ref{app:analytic} discusses the connection to the Takenaka--Malmquist basis.  Appendix~\ref{app:cp-interface} verifies that our filter constructed in Section~\ref{subsec:chen-price-app} can be used in the Chen--Price algorithm.

\section{Preliminaries}\label{sec:prelim}

\paragraph{Fourier-sparse signals and norms.}
For $k\ge1$, let
\begin{equation}\label{eq:Tk-def}
        \Tcal_k:=\Bigl\{F(t)=\sum_{j=1}^k a_j e^{i\lambda_j t}:
        a_j\in\C,\ \lambda_j\in\R\Bigr\}
\end{equation}
denote the class of $k$-Fourier-sparse signals.  Repeated frequencies are always merged.  We define the norms:
$$
        \|h\|_{L^2(I)}:=\Bigl(\int_I|h(t)|^2dt\Bigr)^{1/2},
        \qquad
        \|h\|_{L^\infty(I)}:=\sup_{t\in I}|h(t)|,
$$
and the averaged norm:
$$
        \|h\|_{2,T}^2:=\mathbb E_{t\sim[-T,T]}\left[|h(t)|^2\right]
        :=\frac1{2T}\int_{-T}^{T}|h(t)|^2dt .
$$

\paragraph{Hyperbolic notation.}
For $x\ge1$,
$$
        \arcosh x:=\log\bigl(x+\sqrt{x^2-1}\bigr).
$$
We shall use the following endpoint comparison.

\begin{lemma}[Elementary $\arcosh$ comparisons]\label{lem:arcosh}
For $0\le\delta\le1$,
\begin{equation}\label{eq:arcosh-compare}
        \sqrt\delta\le \arcosh(1+\delta)\le \sqrt{2\delta}\le2\sqrt\delta .
\end{equation}
\end{lemma}

\begin{proof}
For the lower bound, $\cosh y\le1+y^2$ for $0\le y\le1$; setting $y=\sqrt\delta$ gives
$\cosh\sqrt\delta\le1+\delta$, and monotonicity of $\arcosh$ gives the claim.  For the upper bound,
$$
        \arcosh(1+\delta)=\int_1^{1+\delta}\frac{du}{\sqrt{u^2-1}}
        \le \int_0^\delta \frac{dv}{\sqrt v}=2\sqrt\delta,
$$
and the sharper $\sqrt{2\delta}$ follows from $\cosh y\ge1+y^2/2$ with $y=\arcosh(1+\delta)$.
\end{proof}

\paragraph{Approximation theory results.}
The Nikolskii inequality is used to prove the outer regime ($x\geq 2$) of our extrapolation bound in Theorem~\ref{thm:main}.
\begin{lemma}[Nikolskii inequality {\cite[Theorem~2.3]{Erdelyi2016}}]\label{lem:nikolskii}
For every $F\in\Tcal_k$,
\begin{equation}\label{eq:nik-01}
        \|F\|_{L^\infty[0,1]}\leq \frac{\pi k}{2}\,\|F\|_{L^2[0,1]}.
\end{equation}
Consequently, for every $g\in\Tcal_k$,
\begin{equation}\label{eq:nik-11}
        \|g\|_{L^\infty[-1,1]}\leq \frac{\pi k}{2\sqrt2}\,\|g\|_{L^2[-1,1]}.
\end{equation}
\end{lemma}

The following standard analytical results are used for the proofs in Section~\ref{subsec:chen-price-app}.

\begin{fact}[Markov brothers' inequality]\label{fac:Markov}
If $p$ is a real polynomial of degree at most $n$, then
\begin{align*}
    \sup_{x\in[-1,1]}|p'(x)|\le n^2\sup_{x\in[-1,1]}|p(x)|.
\end{align*}
\end{fact}

\begin{fact}[Extremal property of Chebyshev polynomials, \cite{BorweinErdelyi1995}]\label{fac:Chebyshev}
If $p$ is a real polynomial of degree at most $n$, then
\begin{align*}
    |p(s)|\le \sup_{t\in [-1,1]}|p(t)|\cdot \,T_n(s)\qquad \forall s\geq 1,
\end{align*}
where $T_n$ is the $n$-th Chebyshev polynomial.  
\end{fact}

\begin{theorem}[Bernstein's inequality {\cite[Chapter~11]{Boas1954}}]\label{thm:bernstein}
Let $H:\R\to\C$ be bounded, and suppose that its Fourier transform $\widehat H(\xi)=\int_\R H(t)e^{-i\xi t}\,dt$ is supported in $[-B,B]$.  Then $H$ is differentiable and
\[
        \sup_{t\in\R}|H'(t)|\leq B\,\sup_{t\in\R}|H(t)| .
\]
\end{theorem}

\section{Extrapolation bound for Fourier-sparse signals}\label{sec:upper}
In this section, we prove the optimal extrapolation bound for Fourier signals (Theorem~\ref{thm:main}). In Section~\ref{sec:elimination}, we use frequency elimination to prove the case when $1\leq x<2$. In Section~\ref{sec:outer}, we use the Tur\'an-type comparison to prove the case when $x\geq 2$. Finally, in Section~\ref{sec:main-proof}, we finish the proof of Theorem~\ref{thm:main}.
\subsection{The inner regime $1\leq x<2$}\label{sec:elimination}
Let $x=1+\delta$ for $\delta\in [0,1)$. Let $\alpha>0$ be a parameter to be chosen later. Recall the \emph{$\alpha$-weighted energy} of a signal defined in~\eqref{eq:tech_overview_def_energy}:
\begin{align}\label{eq:def_alpha_energy}
        w_\alpha(t):=e^{\alpha(t-1)},\qquad \mathcal E_\alpha(f):=\int_{-\infty}^1w_\alpha(t)|f(t)|^2\,dt.
\end{align}
For $\lambda\in\R$, define the frequency elimination operator
\begin{equation}\label{eq:ema}
        (\mathcal A_\lambda f)(t):=\alpha\int_0^\infty e^{-\alpha s}e^{i\lambda s}f(t-s)\,ds.
\end{equation}

\begin{lemma}[One elimination step]\label{lem:ema-step}
Let $g\in {\cal T}_k$, $\lambda\in\R$, and set $b:=\mathcal A_\lambda g$ and $h:=g-b$. The frequencies of $h$ are exactly those of $g$ other than $\lambda$. Moreover,
\begin{align}
        b'&=i\lambda b+\alpha h,\label{eq:elimination-ode}\\
        \bigl(w_\alpha|b|^2\bigr)'&=\alpha w_\alpha\bigl(|g|^2-|h|^2\bigr).\label{eq:elimination-local}
\end{align}
Integrating the second identity gives
\begin{equation}\label{eq:elimination-split}
        \mathcal E_\alpha(g)=\mathcal E_\alpha(h)+\frac{|b(1)|^2}{\alpha}.
\end{equation}
\end{lemma}

\begin{proof}
For a $1$-Fourier-sparse signal $f(t)=e^{i\mu t}$, direct integration gives
\[
        (\mathcal A_\lambda f)(t)
        =\frac{\alpha}{\alpha+i(\mu-\lambda)}e^{i\mu t}.
\]
Thus, the frequency elimination process $g\mapsto h$ multiplies its coefficient at frequency $\lambda_j$ by $i(\lambda_j-\lambda)/(\alpha+i(\lambda_j-\lambda))$. This multiplier vanishes precisely when $\lambda=\lambda_j$, proving the frequency claim.

To obtain the differential equation, write
\[
        b(t)=\alpha\int_{-\infty}^t e^{(i\lambda-\alpha)(t-u)}g(u)\,du.
\]
Differentiating yields 
\begin{align*}
    b'=\alpha g+(i\lambda-\alpha)b,
\end{align*}
which is \eqref{eq:elimination-ode} because $g=b+h$. 

Since $\lambda\in \R$,
\[
        (|b|^2)'=2\Re(b'\overline b)
        =2\alpha\Re(h\overline b)
        =\alpha\bigl(|g|^2-|h|^2-|b|^2\bigr).
\]
Multiplication by $w_\alpha$ and the identity $w_\alpha'=\alpha w_\alpha$ give \eqref{eq:elimination-local}. 

Finally, $b$ is bounded, $w_\alpha(1)=1$, and $w_\alpha(t)\to0$ as $t\to-\infty$. Integrating \eqref{eq:elimination-local} over $(-\infty,1]$ gives \eqref{eq:elimination-split}.
\end{proof}

We now apply the lemma successively to the distinct frequencies of $g$. 
\begin{lemma}[$k$ elimination steps]\label{lem:elimination}
For any $g(t)=\sum_{j=1}^{k}v_je^{i\lambda_jt}\in\Tcal_k$ with distinct frequencies $\lambda_1,\ldots,\lambda_{k}$, set 
\begin{align*}
    f_0:=&~ g\\
    b_j:=&~ \mathcal A_{\lambda_{j+1}}f_j,\qquad f_{j+1}:=f_j-b_j\qquad \forall 0\leq j<k.
\end{align*}
Then $f_k=0$, and
\begin{equation}\label{eq:elimination-coordinates}
        \sum_{j=0}^{k-1}|b_j(1)|^2=\alpha\mathcal E_\alpha(g),
        \qquad
        \mathcal E_\alpha(f_{j+1})\leq\mathcal E_\alpha(f_j)\leq\mathcal E_\alpha(g).
\end{equation}
For every $0\leq j<k$ and every $t\in\R$, the residuals satisfy
\begin{equation}\label{eq:reconstruction-identity}
        f_j(t)=b_j(1)e^{i\lambda_{j+1}(t-1)}+f_{j+1}(t)
        +\alpha\int_1^t e^{i\lambda_{j+1}(t-s)}f_{j+1}(s)\,ds.
\end{equation}
\end{lemma}

\begin{proof}
Lemma~\ref{lem:ema-step} shows inductively that $f_j$ has frequencies $\lambda_{j+1},\dots,\lambda_k$. Hence $f_k=0$. At each step, \eqref{eq:elimination-split} gives
\[
        \mathcal E_\alpha(f_j)-\mathcal E_\alpha(f_{j+1})=\frac{|b_j(1)|^2}{\alpha}\geq 0.
\]
Thus, the $\alpha$-weighted energies decrease, and summing these differences gives \eqref{eq:elimination-coordinates}. 

To reconstruct $f_j(t)$, solve the differential equation with boundary value $b_j(1)$:
\begin{align*}
    b_j(t)'=i\lambda_{j+1}b_j(t)+\alpha f_{j+1}(t),
\end{align*}
which gives
\begin{align*}
    b_j(t)=b_j(1)e^{i\lambda_{j+1}(t-1)}+\alpha \int_1^t e^{i\lambda_{j+1}(t-s)}f_{j+1}(s)\,ds\qquad \forall t\in \R.
\end{align*}
Then, \eqref{eq:reconstruction-identity} follows from $f_j=b_j+f_{j+1}$.
\end{proof}


The $\alpha$-weighted energy in~\eqref{eq:def_alpha_energy} includes the unobserved region $t<-1$. We control its contribution by proving the following concentration result.

\begin{lemma}[$\alpha$-weighted energy concentration]\label{lem:first-moment}
For every $g\in\Tcal_k$ and every $\alpha>0$,
\begin{equation}\label{eq:first-moment}
        \int_{-\infty}^1(1-t)|g(t)|^2w_\alpha(t)\,dt
        \leq\frac{4k}{\alpha}\mathcal E_\alpha(g).
\end{equation}
Consequently, if $\alpha\geq4k$, then
\begin{equation}\label{eq:norm-transfer}
        \mathcal E_\alpha(g)
        \leq2\int_{-1}^1|g(t)|^2w_\alpha(t)\,dt
        \leq2\|g\|_{L^2[-1,1]}^2.
\end{equation}
\end{lemma}

\begin{proof}
Use the frequency elimination sequence from Lemma~\ref{lem:elimination}. Multiply \eqref{eq:elimination-local} at step $j$ by $1-t$ and integrate by parts:
\begin{align*}
    \alpha \int_{-\infty}^1(1-t)w_\alpha(t)
        \left(|f_j(t)|^2-|f_{j+1}(t)|^2\right)\,dt = &~ \int_{-\infty}^1 (1-t)\left(w_\alpha(t) |b_j(t)|^2\right)'\,dt\\
        =&~ \left.(1-t)\left(w_\alpha(t) |b_j(t)|^2\right)\right\rvert_{-\infty}^1+\int_{-\infty}^1 w_\alpha(t) |b_j(t)|^2\,dt\\
    = &~ {\cal E}_\alpha (b_j),
\end{align*}
where the last step follows from $w_\alpha(t)=e^{\alpha(t-1)}\rightarrow 0$ as $t\rightarrow -\infty$. Summing over $j$ from $0$ to $k-1$, and using $f_0=g$, $f_k=0$ gives
\[
        \alpha\int_{-\infty}^1(1-t)w_\alpha(t)|g(t)|^2\,dt
        =\sum_{j=0}^{k-1}\mathcal E_\alpha(b_j).
\]
Since $b_j=f_j-f_{j+1}$, \eqref{eq:elimination-coordinates} gives
\begin{align*}
        \mathcal E_\alpha(b_j)=&~ \int_{-\infty}^1 w_\alpha(t) |b_j(t)|^2\,dt\leq~ \int_{-\infty}^1 w_\alpha(t) \left(2|f_j(t)|^2+2|f_{j+1}(t)|^2\right)\,dt
        \\=&~2\mathcal E_\alpha(f_j)+2\mathcal E_\alpha(f_{j+1})
        \leq~4\mathcal E_\alpha(g).
\end{align*}
Thus, we have
\begin{align}\label{eq:concentration_expectation}
    \alpha\int_{-\infty}^1(1-t)w_\alpha(t)|g(t)|^2\,dt \leq 4k{\cal E}_\alpha(g),
\end{align}
which is \eqref{eq:first-moment}.

We can define a random variable $t$ with probability density proportional to $w_\alpha(t)|g(t)|^2$ on $(-\infty, 1]$, and $X:=1-t$. Then,  \eqref{eq:first-moment} bounds the expectation of $X$:
\begin{align*}
    \mathbb{E}[X]=\frac{\int_{-\infty}^1(1-t)w_\alpha(t)|g(t)|^2\,dt}{\int_{-\infty}^1w_\alpha(t)|g(t)|^2\,dt}\leq \frac{4k}{\alpha}.
\end{align*}
Thus, for any $L>0$, by Markov's inequality,
\begin{align*}
    \Pr[X\geq L]=\frac{\int_{-\infty}^{1-L} w_\alpha(t)|g(t)|^2\,dt}{{\cal E}_\alpha(g)}\leq \frac{4k}{\alpha L}.
\end{align*}
That is,
\begin{align}\label{eq:energy_localization}
    \int_{-\infty}^{1-L} w_\alpha(t)|g(t)|^2\,dt\leq \frac{4k}{\alpha L}{\cal E}_\alpha(g)\qquad \forall L>0.
\end{align}
In particular, 
\begin{align*}
    \int_{-\infty}^{1-8k/\alpha} w_\alpha(t)|g(t)|^2\,dt\leq \frac{1}{2}{\cal E}_\alpha(g),\qquad \int_{1-8k/\alpha}^{1} w_\alpha(t)|g(t)|^2\,dt\geq \frac{1}{2}{\cal E}_\alpha(g).
\end{align*}
It then implies that
\begin{align}\label{eq:energy_local_result}
    {\cal E}_\alpha(g)\leq 2\int_{1-8k/\alpha}^{1} w_\alpha(t)|g(t)|^2\,dt\leq 2\int_{-1}^{1} |g(t)|^2\,dt=2\|g\|_{L^2[-1,1]}^2,
\end{align}
where the second step follows from $w_\alpha(t)\leq 1$ for $t\leq 1$ and $\alpha \geq 4k$. This proves \eqref{eq:norm-transfer}.
\end{proof}

\begin{remark}\label{rem:localization}
Inequality~\eqref{eq:energy_localization} states that all but a small fraction of the $\alpha$-weighted energy of $g$ lies within a distance of $O(k/\alpha)$ of the endpoint. Inequalities of this type are called infinite--finite range inequalities.  Erd\'elyi~\cite[Theorem~3.1]{Erdelyi2016} proved one for a larger class of exponential sums. In our notation, it states that for every $g\in\Tcal_k$,
\[
    \mathcal E_\alpha(g)\leq\bigl(1+8190e^{-k/10}\bigr)\int_{1-9k/\alpha}^{1}|g(t)|^2e^{\alpha(t-1)}\,dt,
\]
which can also give~\eqref{eq:energy_local_result}.
\end{remark}

We next use \eqref{eq:reconstruction-identity} to recover the original signal from $b_0(1),\dots,b_{k-1}(1)$. 
\begin{definition}[Ordinary Laguerre polynomials]\label{def:def_laguerre}
For any $m\geq 0$, the ordinary Laguerre polynomial $L_m(y)$ is defined as
\begin{align*}\label{eq:laguerre-def}
    L_m(y):=\sum_{\ell=0}^m\binom m\ell\frac{(-y)^\ell}{\ell!}.
\end{align*}
The following lower and upper bounds hold (see, e.g., \cite[Chapter~V]{Szego1975}): 
\begin{equation}\label{eq:laguerre-neg}
        1\leq L_m(-y)=\sum_{\ell=0}^m\binom m\ell\frac{y^\ell}{\ell!}\leq e^{2\sqrt{my}},
        \qquad y\geq0.
\end{equation}
\end{definition}

\begin{lemma}[Signal reconstruction]\label{lem:reconstruction}
Let $b_0(1),\dots,b_{k-1}(1)$ be defined as in Lemma~\ref{lem:elimination}. Then, for every $r\geq0$,
\begin{equation}\label{eq:exterior-majorant}
        |g(1+r)|\leq\sum_{m=0}^{k-1}|b_m(1)|\,\Phi_m(r),
\end{equation}
where
\[
        \Phi_0(r)=1,
        \qquad \Phi_{m+1}(r)=\Phi_m(r)+\alpha\int_0^r\Phi_m(s)\,ds.
\]
In particular,
\begin{equation}\label{eq:Phi-formula}
        \Phi_m(r)=\sum_{j=0}^m\binom mj\frac{(\alpha r)^j}{j!}=L_m(-\alpha r).
\end{equation}
\end{lemma}

\begin{proof}
By~\eqref{eq:reconstruction-identity} with $t=1+r$ and the triangle inequality, we have
\begin{align}\label{eq:reconstruct_triangle_ineq}
        |f_j(1+r)|\leq|b_j(1)|+|f_{j+1}(1+r)|
        +\alpha\int_0^r|f_{j+1}(1+s)|\,ds.
\end{align}
We prove by downward induction that
\[
        |f_j(1+r)|\leq\sum_{m=j}^{k-1}|b_m(1)|\,\Phi_{m-j}(r)
        \qquad\forall r\geq0.
\]
The claim holds for $j=k$, since $f_k=0$. Substituting the bound for $f_{j+1}$ into~\eqref{eq:reconstruct_triangle_ineq} yields
\begin{align*}
        |f_j(1+r)|
        &\leq|b_j(1)|+\sum_{m=j+1}^{k-1}|b_m(1)|
        \left(\Phi_{m-j-1}(r)+\alpha\int_0^r\Phi_{m-j-1}(s)\,ds\right)\\
        &=|b_j(1)|+\sum_{m=j+1}^{k-1}|b_m(1)|\Phi_{m-j}(r)=\sum_{m=j}^{k-1}|b_m(1)|\,\Phi_{m-j}(r).
\end{align*}
Taking $j=0$ proves \eqref{eq:exterior-majorant}.

To solve the recurrence of $\Phi_m$, let 
\begin{align*}
    ({\cal T}\varphi)(r):=\int_0^r\varphi(s)\,ds.
\end{align*}
Then,
\begin{align*}
    \Phi_m(r)=\Phi_{m-1}(r)+\alpha\int_0^r\Phi_{m-1}(s)\,ds=\left((\operatorname{Id}+\alpha {\cal T})\Phi_{m-1}\right)(r).
\end{align*}
Thus,  
\begin{align*}
    \Phi_m(r)=\left((\operatorname{Id}+\alpha {\cal T})^m \Phi_0\right)(r)=\sum_{j=0}^m \binom{m}{j}\alpha^j \left({\cal T}^j\Phi_0\right)(r)=\sum_{j=0}^m \binom{m}{j}\alpha^j \frac{r^j}{j!},
\end{align*}
which is \eqref{eq:Phi-formula} by the definition of the Laguerre polynomial $L_m$.
\end{proof}

Now, we can finish the proof for the inner regime where $1\leq x< 2$.

\begin{proposition}[Extrapolation bound for small $x$]\label{prop:edge}
For every $g\in\Tcal_k$ and every $\delta\geq0$, 
\begin{equation}\label{eq:edge-sharp}
        |g(1+\delta)|\leq2\sqrt{2}ke^{4k\sqrt{\delta}}\,\|g\|_{L^2[-1,1]}.
\end{equation}
\end{proposition}

\begin{proof}
Wlog, we can assume that $g$ has $k$ distinct frequencies. Let $b_0(1),\ldots,b_{k-1}(1)$ be as in Lemma~\ref{lem:elimination}.  

By~\eqref{eq:exterior-majorant} and Cauchy--Schwarz,
\begin{align*}\label{eq:edge-chain}
    |g(1+\delta)|^2\leq&~ \left(\sum_{m=0}^{k-1}|b_m(1)|^2\right)\left(\sum_{m=0}^{k-1}\Phi_m(\delta)^2\right)=\left(\sum_{m=0}^{k-1}|b_m(1)|^2\right)\left(\sum_{m=0}^{k-1}L_m(-\alpha\delta)^2\right)\\
    = &~ \alpha{\cal E}_\alpha(g)\left(\sum_{m=0}^{k-1}L_m(-\alpha\delta)^2\right)\leq \alpha{\cal E}_\alpha(g)\left(\sum_{m=0}^{k-1}e^{4\sqrt{m\alpha \delta}}\right)\\
    \leq &~ \alpha k {\cal E}_\alpha(g)e^{4\sqrt{k\alpha \delta}}\leq 2\alpha k e^{4\sqrt{k\alpha \delta}} \|g\|_{L^2[-1,1]}^2,
\end{align*}
where the third step follows from~\eqref{eq:elimination-coordinates}, the fourth step follows from~\eqref{eq:laguerre-neg}, and the last step follows from~\eqref{eq:norm-transfer} and $\alpha \geq 4k$.

Hence, by taking $\alpha=4k$, we get that
\begin{align*}
    |g(1+\delta)|\leq 2\sqrt{2}ke^{4k\sqrt{\delta}}\,\|g\|_{L^2[-1,1]}.
\end{align*}
The proposition is then proved.
\end{proof}

\subsection[The outer regime]{The outer regime $x\geq2$}\label{sec:outer}

For $x\geq2$, a
coefficient-explicit Tur\'an-type extrapolation is already on the Chebyshev scale, because
$\arcosh x\asymp\log x$ there.

\begin{lemma}[Tur\'an's first main theorem]\label{lem:turan-finite}
Let $z_1,\ldots,z_k\in\C$ satisfy $|z_j|\geq1$,  $b_1,\ldots,b_k\in\C$, and
$s_n=\sum_{j=1}^kb_jz_j^n$ for integers $n\geq0$.  Then for every integer $m\geq0$,
\begin{equation}\label{eq:turan}
        |s_0|\leq k2^{k-1}\binom{m+k-1}{k-1}\max_{m+1\leq n\leq m+k}|s_n| .
\end{equation}
For $m=0$, the prefactor can be replaced by $2^k-1$.
\end{lemma}

\begin{remark}
Lemma~\ref{lem:turan-finite} is classical: it is a weak form of Tur\'an's first main theorem, whose sharp
version, due independently to Makai and to de Bruijn, has the smaller constant
$\sum_{j=0}^{k-1}\binom{m+j}{j}2^j$ in place of $k2^{k-1}\binom{m+k-1}{k-1}$ (and agrees with the constant
$2^k-1$ at $m=0$); see \cite{Turan1984} or \cite[Chapter~5]{Montgomery1994}.
\end{remark}
\begin{proof}
Let
$$
        a(z):=\prod_{j=1}^k\Bigl(1-\frac z{z_j}\Bigr)=\sum_{i=0}^ka_iz^i .
$$
Expanding the product, we have $a_i=(-1)^ie_i(z_1^{-1},\ldots,z_k^{-1})$,
where
\begin{align*}
    e_i(x_1,\dots,x_k)=\sum_{S\in \binom{[k]}{i}}\prod_{j\in S}x_j
\end{align*}
is the $i$-th elementary symmetric polynomial. Since $|z_j|\geq 1$, we have for $0\leq i\leq k$,
\begin{equation}\label{eq:a-coeff}
        |a_i|\leq \sum_{S\in \binom{[k]}{i}}\prod_{j\in S}\left|z_j^{-1}\right|\leq \binom ki\,.
\end{equation}

For any $z\in \C$, $|z|<1\leq\min_j|z_j|$, $a(z)^{-1}$ has the following series expansion:
$$
        \frac1{a(z)}=\prod_{j=1}^k\frac1{1-z/z_j}
        =\prod_{j=1}^k\sum_{n=0}^\infty\Bigl(\frac z{z_j}\Bigr)^n
        =\sum_{\ell=0}^\infty h_\ell\,z^\ell,
        \qquad
        h_\ell=h_\ell(z_1^{-1},\ldots,z_k^{-1})\,,
$$
where 
\begin{align*}
h_\ell(x_1,\dots,x_k)=\sum_{\substack{n_1+\cdots+n_k=\ell\\n_i\geq 0}} x_1^{n_1}\cdots x_k^{n_k}
\end{align*}
is the complete homogeneous symmetric polynomial of degree $\ell$.  The number
of monomials in $h_\ell$ is
$\binom{\ell+k-1}{k-1}$, and each $|z_j^{-1}|\leq 1$.  Hence, for any $\ell\geq 0$,
\begin{equation}\label{eq:h-coeff}
        |h_\ell|\leq\binom{\ell+k-1}{k-1}\,.
\end{equation}

Let $H_m(z):=\sum_{\ell=0}^mh_\ell z^\ell$ be the $m$-th partial sum of the series expansion of $a(z)^{-1}$, and
$$
        \varphi(z):=1-a(z)H_m(z)=\sum_{n}\varphi_n z^n \,.
$$
Since $H_m\equiv1/a\pmod{z^{m+1}}$ (as formal power series), we get $a\,H_m\equiv1\pmod{z^{m+1}}$, i.e.,  all coefficients of $\varphi$ of degree $\leq m$ vanish.  Also
$\deg\varphi\leq\deg a+\deg H_m\leq k+m$, and $\varphi(z_j)=1-a(z_j)H_m(z_j)=1$ because $a(z_j)=0$.
Therefore,
\begin{equation}\label{eq:s0-expansion}
        s_0=\sum_{j=1}^kb_j=\sum_{j=1}^kb_j\varphi(z_j)
        =\sum_{j=1}^kb_j\sum_{n=m+1}^{m+k}\varphi_n z_j^n
        =\sum_{n=m+1}^{m+k}\varphi_n s_n \,.
\end{equation}

It remains to bound the coefficients of $\varphi$.  For $n\geq m+1$,
$$
        \varphi_n=-\sum_{\substack{i+\ell=n\\ 1\leq i\leq k,\ 0\leq\ell\leq m}}a_ih_\ell.
$$
Summing over the relevant range of $n$ and grouping
by $i$: for fixed $i\in\{1,\ldots,k\}$, the constraint $0\leq\ell=n-i\leq m$ together with
$m+1\leq n\leq m+k$ forces $n\in\{m+1,\ldots,m+i\}$, i.e., at most $i$ values of $n$.  Using
\eqref{eq:a-coeff}, \eqref{eq:h-coeff}, and the monotonicity
$\binom{\ell+k-1}{k-1}\leq\binom{m+k-1}{k-1}$ for $\ell\leq m$, we obtain that
$$
        \sum_{n=m+1}^{m+k}|\varphi_n|
        \leq\sum_{i=1}^k i\binom ki\cdot\binom{m+k-1}{k-1}
        =k2^{k-1}\binom{m+k-1}{k-1}\,.
$$
Combined with \eqref{eq:s0-expansion} and the triangle inequality, this proves \eqref{eq:turan}.

If $m=0$, then $H_0=h_0=1$ and $\varphi=1-a$ has coefficients $-a_1,\ldots,-a_k$, so
$\sum_{\nu=1}^k|\varphi_\nu|\leq\sum_{i=1}^k\binom ki=2^k-1$.
\end{proof}

\begin{proposition}[Outer extrapolation]\label{prop:outer}
For every $g\in\Tcal_k$ and every $x\geq2$,
\begin{equation}\label{eq:outer}
        |g(x)|\leq k\exp\left(3k\,\arcosh x\right)\,\|g\|_{L^2[-1,1]} \,.
\end{equation}
\end{proposition}

\begin{proof}
We prove an intermediate bound: for all $x\geq1$ and all $k\geq1$, 
\begin{equation}\label{eq:outer-intermediate}
        |g(x)|\leq k^2\bigl(3e(x+1)\bigr)^k\,\|g\|_{L^2[-1,1]} \,.
\end{equation}
Note that for $x\geq 2$, $\arcosh(x)\geq \log(1+x)$. 
Then, using $\log k\leq k/2$ and $3/2 + \log 3 <2\arcosh (2)\leq 2\arcosh (x)$ for $x\geq 2$, we have
\begin{align*}
    k^2\bigl(3e(x+1)\bigr)^k= &~ k \exp(\log k + k(1+\log 3 + \log(x+1)))\\
    \leq &~ k \exp\left(k\left(\frac{3}{2}+\log 3 + \arcosh (x)\right)\right)\\
    \leq &~ k\exp\left(3k\arcosh(x)\right).
\end{align*}
Substituting into~\eqref{eq:outer-intermediate} gives~\eqref{eq:outer}. 

It remains to prove~\eqref{eq:outer-intermediate}. For $k=1$, $g(t)=v_1e^{i\lambda_1t}$, so $|g(x)|=|v_1|$ while
$\|g\|_{L^2[-1,1]}=\sqrt2\,|v_1|$; thus $|g(x)|\leq\|g\|_{L^2[-1,1]}$ and
\eqref{eq:outer-intermediate} holds trivially.  

Assume $k\geq2$. To apply Tur\'an's first main theorem (Lemma~\ref{lem:turan-finite}), we construct the extrapolation nodes as follows:
$$
        h=\frac2{k+1},\qquad
        m=\max\left(0,\left\lceil\frac{x-1}h\right\rceil-1\right),
        \qquad
        t_n:=x-n h\quad(n\geq0).
$$
We claim that the $k$ consecutive nodes $t_{m+1},\ldots,t_{m+k}$ lie in $[-1,1]$, and
\begin{equation}\label{eq:m-bound}
        m\leq (k+1)(x-1)/2 .
\end{equation}
We prove the claim by considering $m\geq 1$ and $m=0$ separately:
\begin{itemize}
\item If $m\geq1$, then $m+1=\lceil(x-1)/h\rceil$, so 
\begin{align}\label{eq:m_range}
    (x-1)/h\leq m+1<(x-1)/h+1\,.
\end{align}
The first inequality in~\eqref{eq:m_range} gives
\begin{align*}
    t_{m+1}=x-(m+1)h\leq x-(x-1)=1\,.
\end{align*}
And the second inequality in~\eqref{eq:m_range} gives $t_{m+1}>1-h$, which implies that
$$
        t_{m+k}=t_{m+1}-(k-1)h>1-kh=1-\frac{2k}{k+1}=\frac{1-k}{k+1}\geq-1 \,.
$$
We also have
\begin{align*}
    m\leq\lceil(x-1)/h\rceil-1<(x-1)/h=(k+1)(x-1)/2\,.
\end{align*}
\item If $m=0$, then $\lceil(x-1)/h\rceil\leq1$, i.e.\ $x\leq1+h$, so $t_1=x-h\leq1$. And since $x\geq 1$, we have $t_k=x-kh\geq 1-kh \geq-1$ as before. And \eqref{eq:m-bound} trivially holds for $m=0$.
\end{itemize}
Let $g(t)=\sum_{j=1}^kv_je^{i\lambda_jt}$ and set
$$
        z_j:=e^{-i\lambda_jh},\qquad b_j:=v_je^{i\lambda_jx} .
$$
Then $|z_j|=1$ and
$$
        s_n=\sum_{j=1}^kb_jz_j^n=\sum_{j=1}^kv_je^{i\lambda_j(x-n h)}=g(t_n)\,.
$$
In particular, $s_0=g(x)$, and by the claim, for $m+1\leq n \leq m+k$, $|s_n|\leq\|g\|_{L^\infty[-1,1]}$.
Now, we can apply Lemma~\ref{lem:turan-finite} and get that
\begin{equation}\label{eq:outer-raw}
        |g(x)|\leq k2^{k-1}\binom{m+k-1}{k-1}\,\|g\|_{L^\infty[-1,1]}\,.
\end{equation}
We bound the binomial coefficient as follows:
\begin{equation}\label{eq:binom-bound}
\begin{aligned}
    \binom{m+k-1}{k-1} \leq &~ \left(e\frac{m+k-1}{k-1}\right)^{k-1}
    \leq  \left(e\frac{\frac{(k+1)(x-1)}{2}+k-1}{k-1}\right)^{k-1}\\
    \leq &~ \left(e\frac{\frac{3(k-1)(x-1)}{2}+k-1}{k-1}\right)^{k-1} \leq \left(\frac{3e(x+1)}{2}\right)^{k-1}\,,
\end{aligned}
\end{equation}
where the second step follows from~\eqref{eq:m-bound}, and the third step follows from $k\geq 2$. 
Combining \eqref{eq:outer-raw}, \eqref{eq:binom-bound} gives
\begin{align*}
    |g(x)|\leq k2^{k-1} \left(\frac{3e(x+1)}{2}\right)^{k-1}\, \|g\|_{L^\infty[-1,1]} = k (3e(x+1))^{k-1} \, \|g\|_{L^\infty[-1,1]}\,.
\end{align*}
Then, we apply the Nikolskii inequality (Lemma~\ref{lem:nikolskii}):
$$
        |g(x)|\leq k (3e(x+1))^{k-1}\cdot\frac{\pi k}{2\sqrt2}\,\|g\|_{L^2[-1,1]}
        =k^2\,\frac{\pi}{2\sqrt2}\,\bigl(3e(x+1)\bigr)^{k-1}\,\|g\|_{L^2[-1,1]}\,.
$$
Since $\frac\pi{2\sqrt2}<1.2<3e(x+1)$, this is at most $k^2(3e(x+1))^k\|g\|_{L^2[-1,1]}$, which proves the intermediate bound
\eqref{eq:outer-intermediate}.
\end{proof}

\subsection{Proof of the main theorem}\label{sec:main-proof}
\begin{proof}[Proof of Theorem~\ref{thm:main}]
For $k=1$, $g(t)=v_1e^{i\lambda_1t}$ and $|g(x)|=|v_1|=2^{-1/2}\|g\|_{L^2[-1,1]}$, so both estimates hold
trivially.  Assume $k\geq2$.

If $1\leq x\leq2$, write $x=1+\delta$ with $0\leq\delta\leq1$.  Proposition~\ref{prop:edge} gives
\begin{align*}
        |g(1+\delta)|\leq 2\sqrt{2}k\,e^{4k\sqrt\delta}\,\|g\|_{L^2[-1,1]}
        \leq2\sqrt{2} k\,e^{4k\arcosh(x)}\,\|g\|_{L^2[-1,1]}\,,
\end{align*}
where the second step uses the bound $\sqrt\delta\leq\arcosh(1+\delta)=\arcosh(x)$ from~\eqref{eq:arcosh-compare}.  The first inequality above proves~\eqref{eq:main-edge}.

If $x\geq2$, Proposition~\ref{prop:outer} gives
$$
        |g(x)|\leq ke^{3k\arcosh(x)}\,\|g\|_{L^2[-1,1]}\,.
$$
Thus, for $x\geq 1$,
\begin{align*}
    |g(x)|\leq \max\left\{2\sqrt{2}k\,e^{4k\arcosh(x)},ke^{3k\arcosh(x)}\right\}\,\|g\|_{L^2[-1,1]}
    \leq 2\sqrt{2}ke^{4k\arcosh(x)}\,\|g\|_{L^2[-1,1]}\,,
\end{align*}
which proves \eqref{eq:main-global}.
\end{proof}

\section{Proof of the lower bound}\label{sec:sharpness}

We use two ways of turning a polynomial into a Fourier-sparse signal by letting frequencies collide.  The first preserves the bound $|p|\leq1$ on $[-1,1]$ and gives the Chebyshev witness; the second applies to every polynomial.

\begin{lemma}[The Chebyshev witness]\label{lem:witness}
Let $k\geq1$, $0<\varepsilon<\pi/2$, and let $\chi_\varepsilon(t)=T_{k-1}(\sin(\varepsilon t)/\sin\varepsilon)$.  Then:
\begin{enumerate}[label=(\alph*)]
\item $\chi_\varepsilon\in\Tcal_k$: its frequencies lie in the $k$-element set $\{j\varepsilon:|j|\leq {k-1},\ j\equiv {k-1}\pmod 2\}$;
\item $|\chi_\varepsilon(t)|\leq1$ for $t\in[-1,1]$;
\item for every $R\geq1$, $\sup_{|t|\leq R}|\chi_\varepsilon(t)-T_{k-1}(t)|\to0$ as $\varepsilon\to0$.
\end{enumerate}
\end{lemma}

\begin{proof}
(a) The Chebyshev polynomial $T_{k-1}$ has the parity of ${k-1}$: it is a combination of the powers $x^j$ with $j\leq {k-1}$ and $j\equiv {k-1}\pmod2$.  For each such $j$,
\[
        \sin(\varepsilon t)^j=(2i)^{-j}\sum_{\ell=0}^{j}\binom j\ell(-1)^{j-\ell}e^{i(2\ell-j)\varepsilon t}
\]
uses only frequencies $(2\ell-j)\varepsilon$ with $|2\ell-j|\leq {k-1}$ and $2\ell-j\equiv {k-1}\pmod2$.  There are $k$ such frequencies.

(b) For $|t|\leq1$ we have $|\varepsilon t|\leq\varepsilon<\pi/2$, so $|\sin(\varepsilon t)|\leq\sin\varepsilon$.  The argument of $T_{k-1}$ therefore lies in $[-1,1]$, where $|T_{k-1}|\leq1$.

(c) For $|t|\leq R$, using $|\sin y-y|\leq|y|^3/6$ and $\sin\varepsilon\geq2\varepsilon/\pi$,
\[
        \left|\frac{\sin(\varepsilon t)}{\sin\varepsilon}-t\right|
        =\frac{|(\sin(\varepsilon t)-\varepsilon t)-t(\sin\varepsilon-\varepsilon)|}{\sin\varepsilon}
        \leq\frac{\varepsilon^3(R^3+R)/6}{2\varepsilon/\pi}
        =\frac{\pi\varepsilon^2(R^3+R)}{12}.
\]
So $\sin(\varepsilon t)/\sin\varepsilon\to t$ uniformly on $[-R,R]$, with values in a fixed bounded set on which $T_{k-1}$ is uniformly continuous.
\end{proof}

The sine substitution uses at most $k$ frequencies for polynomials of definite parity; a mixture of even and odd terms can produce $2k-1$ frequencies. For general polynomials, the following exponential substitution preserves the $k$-term representation and converges uniformly on compact intervals.

\begin{lemma}[Exponential substitution]\label{lem:exp-substitution}
Let $p$ be a polynomial of degree at most $k-1$, let $0<\varepsilon\leq1$, and let $w_\varepsilon(t):=(e^{i\varepsilon t}-1)/(i\varepsilon)$.  Then $p(w_\varepsilon(t))\in\Tcal_k$ with frequencies in $\{0,\varepsilon,\ldots,(k-1)\varepsilon\}$, and $p(w_\varepsilon)\to p$ uniformly on every bounded interval as $\varepsilon\to0$.
\end{lemma}

\begin{proof}
For $0\leq m\leq k-1$,
\[
        w_\varepsilon(t)^m=(i\varepsilon)^{-m}\sum_{\ell=0}^m\binom m\ell(-1)^{m-\ell}e^{i\ell\varepsilon t}\in\operatorname{span}\left\{e^{i\ell\varepsilon t}:0\leq\ell\leq k-1\right\}.
\]
For real $|t|\leq R$, the estimate $|e^{iy}-1-iy|\leq y^2/2$ gives $|w_\varepsilon(t)-t|\leq\varepsilon t^2/2\leq\varepsilon R^2/2$, and $|w_\varepsilon(t)|=|2\sin(\varepsilon t/2)|/\varepsilon\leq|t|$.  So $w_\varepsilon\to t$ uniformly on $[-R,R]$, with values in the disk of radius $R$, on which $p$ is uniformly continuous.
\end{proof}

The following proposition proves that the exponent $\Theta(k\arcosh x)$ in Theorem~\ref{thm:main} is sharp up to absolute constants and polynomial factors in $k$.

\begin{proposition}[Chebyshev lower bound; cf.\ {\cite[Remark~8.6]{ChenPrice2019}}]\label{prop:sharpness}
For every $k\geq1$ and every $x\geq1$,
$$
        \sup_{0\neq g\in\Tcal_k}
        \frac{|g(x)|}{\|g\|_{L^2[-1,1]}}
        \geq
        \frac{1}{2\sqrt2}
        \exp\!\left((k-1)\arcosh x\right).
$$

\end{proposition}

\begin{proof}
Let $n=k-1$.  By Lemma~\ref{lem:witness}, for every $0<\varepsilon<\pi/2$ the witness $\chi_\varepsilon$ is a nonzero element of $\Tcal_k$ with $\|\chi_\varepsilon\|_{L^2[-1,1]}\leq\sqrt2$, and $\chi_\varepsilon(x)\to T_n(x)$ as $\varepsilon\to0$.  Hence
$$
        \sup_{0\neq g\in\Tcal_k}
        \frac{|g(x)|}{\|g\|_{L^2[-1,1]}}
        \geq
        \lim_{\varepsilon\to0}\frac{|\chi_\varepsilon(x)|}{\sqrt2}
        =\frac{\cosh(n\arcosh x)}{\sqrt2}
        \geq
        \frac{1}{2\sqrt2}
        \exp\!\left(n\arcosh x\right).
$$
\end{proof}

\section{Algorithmic applications}
\label{sec:applications}

This section gives three applications of Theorem~\ref{thm:main}.  Section~\ref{subsec:chen-price-app} improves the resolution of the Chen--Price cluster-location algorithm from $\Delta+\widetilde O(k^3/T)$ to $\Delta+O(k^2/T)$.  Section~\ref{subsec:fourier-interpolation-app} proves the growth conjecture of~\cite{CaiChenShaoWang2026}, making their $\widetilde O(k^3)$-sample interpolation algorithm unconditional.  Section~\ref{subsec:active-learning-app} bounds how much the prediction error of a sparse Fourier model fitted on $[-1,1]$ can grow just outside $[-1,1]$, both at a single point and over an interval, and shows that these bounds are sharp.

\subsection{Improved clustered-frequency estimation}\label{subsec:chen-price-app}
Suppose an unknown signal $g$ satisfies $\operatorname{supp}\widehat g\subseteq[f_0-\Delta,f_0+\Delta]$ with unknown center $f_0\in[-F,F]$, and the algorithm observes samples $y(t)=g(t)+\zeta(t)$, $t\in[-T,T]$, where $\|\zeta\|_{2,T}^2\leq\gamma\,\|g\|_{2,T}^2$ for a sufficiently small absolute constant $\gamma$.  Here
\[
        \|g\|_{2,T}^2:=\frac1{2T}\int_{-T}^{T}|g(t)|^2\,dt
\]
denotes the average energy on $[-T,T]$. For $k$-sparse signals, Chen and Price~\cite{ChenPrice2019} give an algorithm that uses
\[
        O\Bigl(k\log^2k\cdot\log\frac{F}{\Delta_{\mathrm{CP}}\,\delta_{\mathrm{fail}}}\Bigr)
        \ \text{samples, where}\ 
        \Delta_{\mathrm{CP}}:=\Delta+\widetilde O\bigl(k^3/T\bigr),
\]
and with probability at least $1-\delta_{\mathrm{fail}}$ returns $\widehat f_0$ with $|\widehat f_0-f_0|=O(\Delta_{\mathrm{CP}})$.  The term $\widetilde O(k^3/T)$ comes from the bandwidth of the filter that their algorithm applies to the samples.

Our extrapolation bound can be used to improve the filter.  Rescaling Theorem~\ref{thm:main} to the interval $[-T,T]$ gives, for every $g\in\Tcal_k$,
\begin{align*}
        |g(Tx)|&\leq C_0k\,e^{4k\arcosh x}\,\|g\|_{2,T}, && x\geq1,\\
        |g(T(1+\delta))|&\leq C_0k\,e^{4k\sqrt\delta}\,\|g\|_{2,T}, && 0\leq\delta\leq1,
\end{align*}
for an absolute constant $C_0$.  The next theorem constructs a filter of bandwidth $O(k^2/T)$ directly against these bounds.  It multiplies powers of the $\sinc$ function at several scales, extending the single-scale filters of~\cite{CKPS2016,ChenPrice2019}.  Throughout this section, we use the following Fourier transform convention $\widehat f(\xi):=\int_\R f(t)e^{-i\xi t}\,dt$.

\begin{theorem}[Bandwidth-optimal truncation filter]\label{thm:cheb-matched-filter}
For every $0<\eta<1$, there is a constant $C_\eta$ such that the following holds.  For every integer
$k\ge1$ and every $T>0$, there exists a real even function $H_{k,T}$ satisfying
$$
        0\le H_{k,T}(t)\le1,
        \qquad
        \operatorname{supp}\widehat {H_{k,T}}\subseteq[-C_\eta k^2/T,\,C_\eta k^2/T]\,,
$$
and, for every $g\in\Tcal_k$,
\begin{equation}
\label{eq:cheb-filter-energy}
        \int_{-T}^{T}|1-H_{k,T}(t)|^2|g(t)|^2\,dt
        +\int_{\R\setminus[-T,T]}|H_{k,T}(t)|^2|g(t)|^2\,dt
        \le\eta\int_{-T}^{T}|g(t)|^2\,dt\, .
\end{equation}
\end{theorem}

\begin{proof}
It suffices to prove the case $T=1$: if $H_{k,1}$ has been constructed, set
$H_{k,T}(t):=H_{k,1}(t/T)$, so that
$\operatorname{supp}\widehat{H_{k,T}}=T^{-1}\operatorname{supp}\widehat{H_{k,1}}$, and apply the $T=1$
estimate to $u\mapsto g(Tu)\in\Tcal_k$; changing variables back yields \eqref{eq:cheb-filter-energy}.

We prove the $T=1$ case in the following six steps.

\noindent \textbf{Step 1: Construct a multiscale sinc-product kernel.}  Fix $A\ge1$ (an absolute constant chosen later).  For convenience, we will use $C_A$ to denote all constant factors that depend only on $A$ (so that in different equations, the meaning of $C_A$ may differ). Let $B\ge16$.  Let
$$
        J:=\lfloor\log_4B\rfloor,\qquad
        \omega_j:=B4^{-j},\qquad
        m_j:=2\lceil A2^j\rceil\qquad~\forall0\le j\le J\,,
$$
and define the kernel function
\begin{equation}
\label{eq:sinc-product-filter}
        \Psi_{A,B}(u):=\prod_{j=0}^J\left(\sinc(\omega_j u)\right)^{m_j},
        \quad
        Z_{A,B}:=\int_\R\Psi_{A,B}(u)\,du,
        \quad
        K_{A,B}(u):=\frac{\Psi_{A,B}(u)}{Z_{A,B}}\,,
\end{equation}
with the convention $\sinc(0)=1$.  Every $m_j$ is even, so $\Psi_{A,B}\ge0$; the factor $j=0$ decays like
$|u|^{-m_0}$ with $m_0\ge2$, so $\Psi_{A,B}$ is integrable.  Moreover, $\Psi_{A,B}(0)=1$, so
$Z_{A,B}>0$ and $K_{A,B}$ is an even probability density.

Under our Fourier-transform convention,
$$
        \widehat{\sinc(\omega\,\cdot)}(\xi)
        =\frac{\pi}{\omega}\,\mathbf 1_{[-\omega,\omega]}(\xi)
$$
(up to irrelevant endpoint values).  By the convolution theorem of the Fourier transform, $\widehat{(\sinc(\omega\,\cdot))^m}$ is supported in the Minkowski sum of
$m$ copies of $[-\omega,\omega]$, namely $[-m\omega,m\omega]$.  Applying this once more to the product over
$j$,
\begin{align}\label{eq:supp_hat_KAB}
        \operatorname{supp}\widehat {K_{A,B}}\subseteq[-\sigma,\sigma]\,,
\end{align}
where
$$
        \sigma:=\sum_{j=0}^Jm_j\omega_j
        \le\sum_{j=0}^J(2A2^j+2)B4^{-j}\le4AB+\tfrac83B\le7AB \,.
$$
Dividing by the nonzero scalar $Z_{A,B}$ only multiplies the Fourier transform by $Z_{A,B}^{-1}$, and therefore
does not change its support.

We claim that the normalization factor $Z_{A,B}$ satisfies
\begin{equation}
\label{eq:Z-lower-filter}
        Z_{A,B}\ge \frac{2}{e\sqrt{8A}}\frac{1}{B}\,.
\end{equation}
Indeed, for $|z|\le1$, $\sinc(z)>0$ and $\log(\sinc(z))\ge-z^2$, since
\begin{align*}
    \frac{\sin z}z\ge &~ 1-\frac{z^2}6 \qquad \text{and}\qquad e^{-z^2}\le1-z^2+\frac{z^4}2\le1-\frac{z^2}6\,.
\end{align*}
Thus, for $|u|\le1/(\sqrt{8A}\,B)$, where every
$|\omega_ju|=|B4^{-j} u|\le B|u|\le1$, we have
$$
        -\log\left(\Psi_{A,B}(u)\right)\le\sum_{j=0}^Jm_j(\omega_ju)^2
        \le B^2u^2\sum_{j=0}^J(2A2^j+2)16^{-j}\le8AB^2u^2\le1\,.
$$
Hence $\Psi_{A,B}\ge e^{-1}$ on that interval.  Because $\Psi_{A,B}\ge0$ on the whole real line,
we may lower-bound its total integral by integrating over this subinterval alone:
$$
        Z_{A,B}=\int_{\mathbb R}\Psi_{A,B}(u)\,du
        \ge \int_{-1/(\sqrt{8A}B)}^{1/(\sqrt{8A}B)}\Psi_{A,B}(u)\,du
        \ge \frac{2}{e\sqrt{8A}}\frac1B,
$$
which is \eqref{eq:Z-lower-filter}.

\noindent\textbf{Step 2: Estimate the tails of the kernel.}  Let $P_{A,B}(r):=\int_r^\infty K_{A,B}(u)\,du$ for $r\ge0$ be the tail of the kernel.  We claim that there are
an absolute constant $c>0$ and constants $C_A$ (depending only on $A$) such that the following tail bounds hold:
\begin{align}
\label{eq:kernel-far-tail}
        P_{A,B}(r)&\le C_A\,e^{-cA\sqrt B}\,r^{-cA\sqrt B}\,,&&\forall r\ge1\,,\\
\label{eq:kernel-root-tail}
        P_{A,B}(r)&\le C_A\bigl(1+\sqrt{Br}\bigr)e^{-cA\sqrt{Br}}\,,&&\forall 0\le r\le1 \,.
\end{align}
\begin{itemize}
\item \emph{Far tail ($r\geq 1$):}  Let $J':=\lfloor\log_4(B/2)\rfloor\le J$, so that
$\omega_j=B4^{-j}\ge2$ for all $j\le J'$.  For every $u\ge r\ge1$ and $j\le J'$,
\begin{align}\label{eq:bound_sinc}
    |\sinc(\omega_j u)|
    =\left|\frac{\sin(\omega_ju)}{\omega_ju}\right|
    \le \frac{1}{\omega_ju}\le \frac{1}{2u}\,.
\end{align}
For $J'<j\le J$ we use the trivial bound $|\sinc(\omega_j u)|\le1$.  Hence, $ \Psi_{A,B}(u)\le (2u)^{-M}$, where
$$
    M:=\sum_{j=0}^{J'}m_j
    \ge2A\,2^{J'}\ge\frac{A\sqrt B}{2\sqrt2}\,.
$$
By \eqref{eq:Z-lower-filter}, for $r\ge1$,
$$
        P_{A,B}(r)
        \le C_AB\,\frac{2^{-M}}{M-1}\,r^{-(M-1)}.
$$
By choosing $A$ larger than an absolute constant, $M-1\ge
A\sqrt B/(4\sqrt{2})$.  Also $\log B\le\sqrt B$ for $B\ge16$, and therefore
$$
        C_AB\,\frac{2^{-M}}{M-1}
        \le C_Ae^{-c_1A\sqrt B}
$$
for an absolute $c_1>0$, while
$$
        r^{-(M-1)}\le r^{-c_2A\sqrt B}
$$
for $c_2:=1/(4\sqrt{2})$.  Taking $c:=\min\{c_1,c_2\}$ proves
\eqref{eq:kernel-far-tail}.

\item \emph{Near tail ($r\in [0,1]$):}  If $Br\le4$, then \eqref{eq:kernel-root-tail} holds trivially after enlarging $C_A$,
since $P_{A,B}(r)\le1$.  Assume $Br>4$.  For each $u\in[r,1]$, let
$j_*(u):=\lfloor\log_4(Bu/2)\rfloor\in[0,J]$.
Then, $\omega_ju = B4^{-j} u\ge2$ for all $j\le j_*(u)$, and  as in~\eqref{eq:bound_sinc}, we have $|\sinc(\omega_ju)|\le\frac12$ for those $j$. Thus, $\Psi_{A,B}(u)\leq 2^{-M_u}$, where
$$
        M_u:=\sum_{j=0}^{j_*(u)} m_j = \sum_{j=0}^{j_*(u)} 2\lceil A2^j\rceil \geq 2A2^{j_*(u)} \geq 2 A \cdot \frac{\sqrt{Bu}}{2\sqrt2}\geq  cA\sqrt{Bu}\,.
$$
Hence, by \eqref{eq:Z-lower-filter} and the substitution $y=\sqrt{Bu}$ (so $B\,du=2y\,dy$), we have
\begin{align*}
        P_{A,B}(r)
        &\le\frac1{Z_{A,B}}\int_r^1e^{-cA\sqrt{Bu}}\,du+P_{A,B}(1)\\
        &\le C_A\int_{\sqrt{Br}}^\infty ye^{-cAy}\,dy+P_{A,B}(1)\\
        &\le C_A(1+\sqrt{Br})e^{-cA\sqrt{Br}},
\end{align*}
where in the last step, $P_{A,B}(1)\le C_Ae^{-cA\sqrt B}$ by \eqref{eq:kernel-far-tail} is absorbed because $r\le1$.
\end{itemize}

\noindent\textbf{Step 3: Construct the filter.}  Define
\begin{equation}
\label{eq:HAB-def}
        H_{A,B}(t):=(\mathbf 1_{[-1,1]}*K_{A,B})(t)=\int_{-1}^1K_{A,B}(t-s)\,ds .
\end{equation}
Then, $H_{A,B}$ is real and even, $0\le H_{A,B}\le1$, and
$\operatorname{supp}\widehat {H_{A,B}}\subseteq[-7AB,7AB]$, by the identity
$\widehat{H_{A,B}}=\widehat{\mathbf 1_{[-1,1]}}\cdot\widehat{K_{A,B}}$ and~\eqref{eq:supp_hat_KAB}.  Moreover, for $r\ge0$,
\begin{align}\label{eq:H_AB_near_bound}
        H_{A,B}(1+r)=\int_r^{2+r}K_{A,B}(u)\,du\le P_{A,B}(r),
\end{align}
with the same bound at $-1-r$ by symmetry. And for $t\in[-1,1]$, we have
$$
        1-H_{A,B}(t)=P_{A,B}(1-t)+P_{A,B}(1+t)\le2P_{A,B}(1-|t|) .
$$
Consequently, by \eqref{eq:kernel-root-tail} and the substitution $y=\sqrt{Br}$,
\begin{equation}
\label{eq:interior-defect-filter}
\begin{aligned}
        \int_{-1}^1|1-H_{A,B}(t)|^2\,dt\le &~ 8\int_0^1P_{A,B}(r)^2\,dr\\
        \leq &~ 8\int_0^1 C_A^2 \bigl(1+\sqrt{Br}\bigr)^2e^{-2cA\sqrt{Br}}\,dr\\
        \le &~ \frac{C_A}B\int_0^\infty
        y(1+y)^2e^{-2cAy}\,dy\le\frac{C_A}B\,.
\end{aligned}
\end{equation}
Combining with the Nikolskii inequality \eqref{eq:nik-11}, we get that
\begin{equation}
\label{eq:inside-loss-filter}
\begin{aligned}
        \int_{-1}^1|1-H_{A,B}(t)|^2|g(t)|^2\,dt
        \le &~ \|g\|_{L^\infty[-1,1]}^2\int_{-1}^1|1-H_{A,B}(t)|^2\,dt\\
        \leq &~ \|g\|_{L^\infty[-1,1]}^2 \cdot \frac{C_A}{B}\\
        \leq &~ \frac{C_Ak^2}B\,\|g\|_{L^2[-1,1]}^2 \,.
\end{aligned}
\end{equation}

\noindent\textbf{Step 4: Control the outside leakage (near range).}  From now on $B=\beta k^2$ with $\beta\ge16$ to be chosen.
For $0\le r\le1$, Proposition~\ref{prop:edge} gives
\begin{align*}
|g(1+r)|^2\le8k^2e^{8k\sqrt r}\,\|g\|_{L^2[-1,1]}^2\,.
\end{align*}
By~\eqref{eq:H_AB_near_bound} and~\eqref{eq:kernel-root-tail},
\begin{align*}
        \int_0^1|H_{A,B}(1+r)|^2|g(1+r)|^2\,dr\leq &~ \int_0^1|P_{A,B}(r)|^2|g(1+r)|^2\,dr\\
        \leq &~ C_Ak^2\,\|g\|_{L^2[-1,1]}^2 \cdot \int_0^1|P_{A,B}(r)|^2e^{8k\sqrt r}\,dr\\
        \le &~ C_Ak^2\,\|g\|_{L^2[-1,1]}^2
        \int_0^1\bigl(1+\sqrt{Br}\bigr)^2\,e^{-2cA\sqrt{Br}+8k\sqrt r}\,dr \,.
\end{align*}
Since $\sqrt{Br}=\sqrt\beta\,k\sqrt r$, the exponent equals
$-(2cA-8/\sqrt\beta)\sqrt{Br}$.  Fix $A$ so large (absolute) that $cA\ge1$; then for $\beta\ge64$
we have $8/\sqrt\beta\le1\le cA$, so the exponent is at most $-cA\sqrt{Br}$. Then, with
$y=\sqrt{Br}$, we get that
$$
        \int_0^1\bigl(1+\sqrt{Br}\bigr)^2e^{-cA\sqrt{Br}}\,dr
        \le\frac2B\int_0^\infty y(1+y)^2e^{-cAy}\,dy\le\frac{C_A}B \,.
$$
Thus the near right tail can be bounded by:
\begin{align}\label{eq:H_near_tail}
    \int_0^1|H_{A,B}(1+r)|^2|g(1+r)|^2\,dr \leq \frac{C_A}{B}k^2\|g\|_{L^2[-1,1]}^2=\frac{C_A}{\beta}\|g\|_{L^2[-1,1]}^2\,.
\end{align}
And the near left tail is identical by symmetry.

\noindent\textbf{Step 5: Control the outside leakage (far range).}  For $r\ge1$, note that
\begin{align*}
    \arcosh (1+r)\le\log(2(1+r))\le \log ((1+r)^2)\,.
\end{align*}
By \eqref{eq:main-global} squared, we have
\begin{align*}
        |g(1+r)|^2\leq &~ Ck^2\exp\left(8\,k\arcosh (1+r)\right)\|g\|_{L^2[-1,1]}^2\\
        \leq &~ C\,k^2\,(1+r)^{16k}\,\|g\|_{L^2[-1,1]}^2\,,
\end{align*}
while \eqref{eq:H_AB_near_bound} and \eqref{eq:kernel-far-tail} give that
\begin{align*}
    |H_{A,B}(1+r)|^2\leq |P_{A,B}(r)|^2\le C_Ae^{-2cA\sqrt B}r^{-2cA\sqrt B} =C_Ae^{-2cA\sqrt{\beta}k}r^{-2cA\sqrt{\beta}k}\,.
\end{align*}
Using
$(1+r)^{16k}\le2^{16k}r^{16k}$ and requiring $cA\sqrt\beta\ge18$,
\begin{align*}
        \int_1^\infty|H_{A,B}(1+r)|^2|g(1+r)|^2\,dr
        &\le C_Ak^22^{16k}e^{-2cA\sqrt\beta k}
        \int_1^\infty r^{16k-2cA\sqrt\beta k}\,dr\cdot \|g\|^2_{L^2[-1,1]}\\
        &\le C_Ae^{16k-2cA\sqrt\beta k}\,\|g\|_{L^2[-1,1]}^2\,,
\end{align*}
where the second step follows from $k^42^{17k}\le e^{16k}$ for every $k\ge1$ and the integral is at most $1$.  Increase $\beta$ once more so that $(2cA-1)\sqrt\beta\ge16$; since $cA\ge1$, it is enough to
take $\beta\ge256$.  Then, we have $2cA\sqrt\beta-16\ge\sqrt\beta$, and the far right tail can be bounded by:
\begin{align}\label{eq:H_far_tail}
    \int_1^\infty|H_{A,B}(1+r)|^2|g(1+r)|^2\,dr \leq C_Ae^{-\sqrt\beta k}\|g\|_{L^2[-1,1]}^2\le C_Ae^{-\sqrt\beta}\|g\|_{L^2[-1,1]}^2\,.
\end{align}
The far left tail is identical by symmetry.

\noindent\textbf{Step 6: Combine the estimates.}  Set $H_{k,1}:=H_{A,\beta k^2}$.  Summing \eqref{eq:inside-loss-filter}, the two near-tail bounds \eqref{eq:H_near_tail}, and the two far-tail bounds \eqref{eq:H_far_tail}, we have
\begin{align*}
        &~ \int_{-1}^1|1-H_{k,1}|^2|g|^2\,dt+\int_{\R\setminus[-1,1]}|H_{k,1}|^2|g|^2\,dt\\
        \leq &~ \frac{C_A k^2}{\beta k^2}\|g\|_{L^2[-1,1]}^2 + 2\left(\frac{C_A}{\beta}\|g\|_{L^2[-1,1]}^2 + C_Ae^{-\sqrt\beta}\|g\|_{L^2[-1,1]}^2\right)\\
        \le &~ C_A\bigl(\beta^{-1}+e^{-\sqrt\beta}\bigr)\|g\|_{L^2[-1,1]}^2\,.
\end{align*}
We can choose $\beta=\beta(\eta)$ sufficiently large to make the right-hand side at most
$\eta\|g\|_{L^2[-1,1]}^2$, which is \eqref{eq:cheb-filter-energy} for $T=1$, with Fourier support
contained in $[-7A\beta k^2,7A\beta k^2]$, i.e.,  $C_\eta=7A\beta(\eta)$. Therefore, the case $T=1$ is proved.
\end{proof}

The following proposition shows that the bandwidth of the filter in Theorem~\ref{thm:cheb-matched-filter} cannot be improved by more than a constant factor.

\begin{proposition}[Bandwidth lower bound]\label{prop:bandwidth-floor}
Let $k\geq1$ and $T>0$, and let $H:\R\to\C$ satisfy $\sup_\R|H|\leq2$ and $\operatorname{supp}\widehat H\subseteq[-B,B]$.  If \eqref{eq:cheb-filter-energy} holds with $H$ in place of $H_{k,T}$ for every $g\in\Tcal_k$, with some $\eta<1/128$, then $B>k^2/(4T)$.
\end{proposition}

\begin{proof}
Replacing $H$ by $H(T\,\cdot)$ and $g$ by $g(T\,\cdot)$, we may assume $T=1$.  Suppose, for contradiction, that $B\leq k^2/4$.

Let $P_m$ denote the Legendre polynomials, normalized by $P_m(1)=1$, and let
\[
        \kappa(t):=\sum_{m=0}^{k-1}\frac{2m+1}2\,P_m(t).
\]
By orthogonality, $\int_{-1}^1p(t)\kappa(t)\,dt=p(1)$ for every polynomial $p$ of degree less than $k$; that is, $\kappa$ is the reproducing kernel of these polynomials at the endpoint $1$. Then 
\begin{align*}
    \kappa(1)=\|\kappa\|_{L^2[-1,1]}^2=\frac{k^2}{2},\qquad \|\kappa\|_{L^\infty[-1,1]}=\kappa(1),
\end{align*}
where the second equation follows from $|P_m|\leq1$ on $[-1,1]$.  

Let $\delta_0:=1/(4k^2)$, $I_-:=[1-\delta_0,1]$, and $I_+:=[1,1+\delta_0]$. If $k=1$, then $\kappa\equiv\frac12$ and the bounds below are immediate, so assume $k\geq2$.  By the Markov brothers' inequality (Fact~\ref{fac:Markov}), 
\begin{align*}
    |\kappa'(t)|\leq(k-1)^2\kappa(1)\qquad \forall t\in [-1,1].
\end{align*}
Thus, for any $t\in I_-$,
\begin{align*}
    |\kappa(t)-\kappa(1)|\leq \delta_0 \cdot (k-1)^2\kappa(1) \leq \frac{1}{4}\kappa(1), \qquad \kappa(t)\geq \frac{3}{4}\kappa(1)=\frac{3k^2}{8}.
\end{align*}
Since $\kappa'$ is a real polynomial of degree at most $k-2$, by the extremal property of Chebyshev polynomials (Fact~\ref{fac:Chebyshev}), for $s\in I_+$,
\begin{align*}
        |\kappa'(s)|\leq &~ (k-1)^2\kappa(1)\,T_{k-2}(1+\delta_0)\\
        \leq &~ (k-1)^2\kappa(1)\cosh\!\left((k-2)\sqrt{2\delta_0}\right)\\
        \leq &~ \cosh\!\left(1/\sqrt2\right)(k-1)^2\kappa(1).
\end{align*}
Since $\cosh(1/\sqrt2)<1.27$, we have for any $t\in I_+$,
\begin{align*}
    |\kappa(t)|\geq (1-1.27/4)\,\kappa(1)\geq\frac23\kappa(1)=\frac{k^2}3.
\end{align*}
Consequently,
\[
        \int_{I_-}|\kappa(t)|^2\,dt\geq\frac{k^2}{36},\qquad \int_{I_+}|\kappa(t)|^2\,dt\geq\frac{k^2}{36}.
\]

By Bernstein's inequality (Theorem~\ref{thm:bernstein}), 
\begin{align*}
    \sup_{t\in \R}|H'(t)|\leq2B\leq k^2/2,
\end{align*}
which implies that
\begin{align*}
    |H(t)-H(1)|\leq\delta_0\frac{k^2}{2}=\frac{1}{8}\qquad \forall t\in I_-\cup I_+.
\end{align*}
Since $|1-H(1)|+|H(1)|\geq1$, either $|1-H(1)|\geq\frac12$, and then $|1-H(t)|\geq\frac38$ for $t\in I_-$; or $|H(1)|\geq\frac12$, and then $|H(t)|\geq\frac38$ for $t\in I_+$.  In either case,
\[
        \int_{I_-}|1-H(t)|^2|\kappa(t)|^2\,dt+\int_{I_+}|H(t)|^2|\kappa(t)|^2\,dt
        \geq\frac9{64}\cdot\frac{k^2}{36}=\frac{k^2}{256}.
\]

By Lemma~\ref{lem:exp-substitution}, $g_\varepsilon:=\kappa\circ w_\varepsilon\in\Tcal_k$ and $g_\varepsilon\to\kappa$ uniformly on $[-1,2]$ as $\varepsilon\to0$.  Applying \eqref{eq:cheb-filter-energy} to $g_\varepsilon$, discarding all of the left side except the integrals over $I_-$ and $I_+$, and letting $\varepsilon\to0$ gives $k^2/256\leq\eta\,\|\kappa\|_{L^2[-1,1]}^2=\eta k^2/2$, that is, $\eta\geq1/128$, a contradiction.
\end{proof}

\begin{corollary}[Improved clustered-frequency estimation]
\label{cor:chen-price-main-improvement}
In the clustered-frequency model above, for every failure probability
$\delta_{\mathrm{fail}}\in(0,1)$, there is an efficient algorithm using
$$
        O\!\left(
        k\log^2k\cdot
        \log\frac{F}{\Delta_*\delta_{\mathrm{fail}}}
        \right)
$$
samples and returning an estimate $\widehat f_0$ such that, with probability at least
$1-\delta_{\mathrm{fail}}$,
$$
        |\widehat f_0-f_0|\le O(\Delta_*),
        \qquad
        \Delta_*:=\Delta+O\!\left(\frac{k^2}{T}\right).
$$
Thus the additive resolution improves from $\Delta+\widetilde O(k^3/T)$ in
\cite{ChenPrice2019} to $\Delta+O(k^2/T)$, while the sample complexity is unchanged up to the logarithmic factors
already present in their theorem.
\end{corollary}

\begin{proof}
Fix a sufficiently small absolute $\eta$ and use the filter from Theorem~\ref{thm:cheb-matched-filter}.  Appendix~\ref{app:cp-interface} verifies all the conditions for the filter used in  \cite{ChenPrice2019}. As a result, we obtain an analog of
\cite[Lemma~8.2]{ChenPrice2019} with $\Delta_* = \Delta+O(k^2/T)$ and
$O(k\log^2k)$ samples per phase estimate.  Applying the same multiscale frequency-search routine as in
\cite[Theorem~1.5]{ChenPrice2019} gives the stated guarantee and the additional logarithmic factor.
\end{proof}

\subsection{Fourier interpolation from cubically many samples}
\label{subsec:fourier-interpolation-app}

The second application of our extrapolation bound is in Fourier interpolation, following Cai, Chen, Shao, and Wang \cite{CaiChenShaoWang2026}.  An unknown signal $g(t)=\sum_{j=1}^kv_je^{2\pi if_jt}$ has arbitrary frequencies $f_j\in[-F,F]$.  An algorithm queries noisy values $y(t)=g(t)+\zeta(t)$ at chosen times $t\in[-T,T]$, where
\[
        \|\zeta\|_{L^2[-T,T]}^2\leq\gamma\,\|g\|_{L^2[-T,T]}^2
\]
for a sufficiently small constant $\gamma>0$, and must output $\widehat g$ with 
\begin{align*}
    \|\widehat g-g\|_{L^2[-T,T]}^2\leq O(\gamma)\,\|g\|_{L^2[-T,T]}^2.
\end{align*}
The algorithm of \cite[Theorem~1.2]{CaiChenShaoWang2026} uses $m=\widetilde O(k^{15/4})$ samples and $\widetilde O(m^\omega)$ time, and its output has the form $\widehat g(t)=\sum_{j=1}^{O(k)}e^{2\pi i\widehat f_jt}Q_j(t)$ with polynomials $Q_j$.  Their Theorem~1.4 reduces the sample complexity to $\widetilde O(k^3)$ under the following growth hypothesis, which they state for $T=1$.

\begin{conjecture}[{\cite[Conjecture~1.3]{CaiChenShaoWang2026}}]\label{conj:CCSW}
For every $g\in\Tcal_k$ and every $t\in\R\setminus[-1,1]$,
\[
        |g(t)|\leq k^{O(1)}\cdot\max_{s\in[-1,1]}|g(s)|\cdot\min\Bigl\{e^{O(k)\sqrt{|t|-1}},\ O(|t|)^k\Bigr\}.
\]
\end{conjecture}

Corollary~\ref{cor:CCSW-interpolation} verifies both branches of this growth bound with explicit constants. As described in Section~\ref{sec:interpolation_corollary}, the square-root branch improves the localization radius in the analysis of \cite{CaiChenShaoWang2026} from $\widetilde O(k^{11/4}/T)$ to $\widetilde O(k^2/T)$, reducing the sample count from $\widetilde O(k^{15/4})$ to $\widetilde O(k^3)$. Their analysis also allows the slightly weaker near-endpoint bound $e^{\widetilde O(k)\sqrt{|t|-1}}$. The corollary uses their algorithm and analysis with the stated growth hypothesis.

\begin{corollary}[Cubic-sample Fourier interpolation]
\label{cor:CCSW-interpolation}
Conjecture~\ref{conj:CCSW} holds: for every $g\in\Tcal_k$ and every $t$ with $|t|\geq1$,
\begin{equation}\label{eq:CCSW-growth}
        |g(t)|\leq 4k^2\,\|g\|_{L^\infty[-1,1]}\cdot\min\Bigl\{e^{4k\sqrt{|t|-1}},\ (6e|t|)^k\Bigr\}.
\end{equation}
Consequently, the conclusion of \cite[Theorem~1.4]{CaiChenShaoWang2026} holds unconditionally: there is an algorithm that takes $m=\widetilde O(k^3)$ samples and $\widetilde O(m^\omega)=\widetilde O(k^{3\omega})$ arithmetic time and, with constant success probability, outputs 
\begin{align*}
    \widehat g(t)=\sum_{j=1}^{O(k)}e^{2\pi i\widehat f_jt}Q_j(t)
\end{align*}
with 
\begin{align*}
    \|\widehat g-g\|_{L^2[-T,T]}^2\leq O(\gamma)\|g\|_{L^2[-T,T]}^2.
\end{align*}
\end{corollary}

\begin{proof}
Since $t\mapsto g(-t)$ is again in $\Tcal_k$, we may assume $t\geq1$.  We bound $|g(t)|$ by each term of the minimum separately, using $\|g\|_{L^2[-1,1]}\leq\sqrt2\,\|g\|_{L^\infty[-1,1]}$.  For the first term, Proposition~\ref{prop:edge} gives
\[
        |g(t)|\leq2\sqrt2\,k\,e^{4k\sqrt{t-1}}\,\|g\|_{L^2[-1,1]}\leq4k\,e^{4k\sqrt{t-1}}\,\|g\|_{L^\infty[-1,1]}.
\]
For the second term, the intermediate bound~\eqref{eq:outer-intermediate}, valid for all $t\geq1$, gives
\[
        |g(t)|\leq k^2\bigl(3e(t+1)\bigr)^k\|g\|_{L^2[-1,1]}\leq\sqrt2\,k^2(6et)^k\,\|g\|_{L^\infty[-1,1]}.
\]
They together prove~\eqref{eq:CCSW-growth}, which is Conjecture~\ref{conj:CCSW}.  The interpolation guarantee now follows from \cite[Theorem~1.4]{CaiChenShaoWang2026}.  For general $T$, apply~\eqref{eq:CCSW-growth} to $t\mapsto g(Tt)$.
\end{proof}

\subsection{Extrapolative active regression and Fourier-feature leverage scores}
\label{subsec:active-learning-app}
Fix a frequency set
$\Lambda=(\lambda_1,\ldots,\lambda_k)$ and write
$$
        E_\Lambda:=\operatorname{span}\{e^{i\lambda_jt}:1\le j\le k\}.
$$
With Lebesgue measure on $[-1,1]$, the squared norm of the evaluation functional at $x$ is the diagonal value of the reproducing kernel of $E_\Lambda$, that is, the reciprocal of the Christoffel
function:
\begin{align}\label{eq:christoffel-formula}
        \Kcal_\Lambda(x)
        =\sup_{0\ne g\in E_\Lambda}\frac{|g(x)|^2}{\|g\|_{L^2[-1,1]}^2}
        =\nu(x)^*G^{-1}\nu(x)\,,
\end{align}
where
$$
        G_{ab}=\int_{-1}^1 e^{i(\lambda_b-\lambda_a)t}dt=2\sinc(\lambda_b-\lambda_a),
        \qquad
        \nu_a(x)=e^{i\lambda_ax}\,.
$$
Indeed, $c^*Gc=\|\sum_jc_je^{i\lambda_jt}\|_{L^2[-1,1]}^2$ and $\sum_jc_je^{i\lambda_jx}=\nu(x)^\mathsf T c$;
Cauchy--Schwarz in the $G$-inner product gives \eqref{eq:christoffel-formula}, with equality at the kernel vector.
Thus, Theorem~\ref{thm:main} is equivalently an outside-interval bound for $\Kcal_\Lambda$ that is uniform over all frequency sets of size at most $k$.

If the training distribution is the uniform probability measure on $[-1,1]$, the corresponding statistical leverage score is
$\tau_\Lambda(x)=2\Kcal_\Lambda(x)$.  Leverage-score sampling is a standard primitive for subspace embeddings
and active least squares, and Fourier-sparse leverage bounds have been used for continuous sparse Fourier regression
and kernel approximation \cite{ChenPriceActive2019,AKMMVZ2019,ErdelyiMuscoMusco2020}.  Our contribution here is a
uniform bound when the prediction point lies outside the support of the training distribution.

For the nonlinear class $\Tcal_k$, define the frequency-uniform exterior leverage relative to the averaged
training norm by
$$
        \tau_k(x):=\sup_{0\ne g\in\Tcal_k}
        \frac{|g(x)|^2}{\mathbb E_{t\sim[-1,1]}[|g(t)|^2]}\,.
$$
Squaring the endpoint estimate \eqref{eq:main-edge} gives an absolute constant $C$ such that, for
$0\leq\delta\leq1$,
\begin{equation}
\label{eq:outside-sensitivity-app}
        \tau_k(1+\delta)+\tau_k(-1-\delta)
        \leq Ck^2e^{Ck\sqrt\delta}\,.
\end{equation}

For $0\leq\delta\leq1$, let
$$
        \mathcal X_\delta:=[1,1+\delta]\cup[-1-\delta,-1]
$$
be the adjacent exterior prediction region.  Two related quantities arise.  The integrated pointwise leverage envelope is
$$
        \mathfrak L_k(\delta):=\int_{\mathcal X_\delta}\tau_k(x)\,dx = \int_{\mathcal X_\delta}\sup_{0\ne g\in\Tcal_k}
        \frac{|g(x)|^2}{\mathbb E_{t\sim[-1,1]}[|g(t)|^2]}\,dx\,,
$$
whereas the actual worst-case $L^2$ transfer constant is
$$
        \Gamma_k(\delta)
        :=\sup_{0\ne g\in\Tcal_k}
        \frac{\displaystyle\int_{\mathcal X_\delta}|g(x)|^2\,dx}
             {\displaystyle\mathbb E_{t\sim[-1,1]}\left[|g(t)|^2\right]}\,.
$$
The pointwise maximizer in $\tau_k(x)$ may depend on $x$. Applying the pointwise bound to any fixed signal and then integrating gives
$\Gamma_k(\delta)\leq\mathfrak L_k(\delta)$.  The next proposition bounds both and proves sharpness for the
actual transfer constant using a single sparse Fourier sum.

\begin{proposition}[Exterior leverage and transfer at the Chebyshev scale]
\label{prop:exterior-leverage}
There are absolute constants $c,C>0$ such that, for every $k\geq2$ and $0\leq\delta\leq1$,
$$
        \Gamma_k(\delta)
        \leq \mathfrak L_k(\delta)
        \leq C(1+k\sqrt\delta)e^{Ck\sqrt\delta}\,.
$$
Conversely, whenever $k^{-2}\leq\delta\leq1$,
$$
        \Gamma_k(\delta)
        \geq c\delta e^{ck\sqrt\delta}\,.
$$
Thus the exponential dependence $\exp(\Theta(k\sqrt\delta))$ is unavoidable, up to polynomial factors, for the
true worst-case transfer from the training interval to the adjacent exterior region.
\end{proposition}

\begin{proof}
The inequality $\Gamma_k(\delta)\leq\mathfrak L_k(\delta)$ follows by applying the definition of
$\tau_k(x)$ to a fixed $g$ and then integrating.  Integrating \eqref{eq:outside-sensitivity-app} and substituting
$y=k\sqrt u$ gives
$$
        \mathfrak L_k(\delta)
        \leq Ck^2\int_0^\delta e^{Ck\sqrt u}\,du
        =2C\int_0^{k\sqrt\delta}ye^{Cy}\,dy
        \leq C(1+k\sqrt\delta)e^{Ck\sqrt\delta}\,.
$$

For the lower bound, let $n=k-1$ and use the Chebyshev witness $\chi_\varepsilon$ of Lemma~\ref{lem:witness}.  Since $|\chi_\varepsilon|\leq1$ on $[-1,1]$, the denominator in the definition of $\Gamma_k(\delta)$ is at most $1$ for $g=\chi_\varepsilon$, and $\chi_\varepsilon\to T_n$ uniformly on $[-1-\delta,1+\delta]$.  Hence
$$
        \Gamma_k(\delta)
        \geq
        \lim_{\varepsilon\to0}
        \int_{\mathcal X_\delta}|\chi_\varepsilon(x)|^2\,dx
        =
        \int_{\mathcal X_\delta}|T_n(x)|^2\,dx\,.
$$
For $u\in[\delta/2,\delta]$, Lemma~\ref{lem:arcosh} gives
$\arcosh(1+u)\geq\sqrt u\geq\sqrt{\delta/2}$, and therefore
$$
        |T_n(1+u)|^2
        =\cosh^2\!\left(n\arcosh(1+u)\right)
        \geq\frac14\exp\!\left(2n\arcosh(1+u)\right)
        \geq c e^{ck\sqrt\delta}\,.
$$
Integrating over $u\in[\delta/2,\delta]$, an interval of length $\delta/2$, proves
$\Gamma_k(\delta)\geq c\delta e^{ck\sqrt\delta}$.
\end{proof}

In particular, at prediction distance $\delta=\Theta(1)$ the worst-case transfer constant is $\Gamma_k(\delta)=e^{\Theta(k)}$, whereas the elementary bound \eqref{eq:elementary-bound} gives only $\Gamma_k(\delta)\leq e^{O(k^2)}$.  Likewise, the amplification stays polynomial in $k$ for $\delta=O(\log^2k/k^2)$, whereas the elementary bound guarantees this only for $\delta=O(\log k/k^2)$.

The transfer constant gives the precise black-box implication from an in-domain regression guarantee to exterior
prediction.

\begin{corollary}[Black-box transfer from active regression to exterior prediction]
\label{cor:active-exterior-transfer}
Fix $|\Lambda|\leq k$, and let $g_\star,\widehat g\in E_\Lambda$.  Then
$$
        \int_{\mathcal X_\delta}|\widehat g(x)-g_\star(x)|^2\,dx
        \leq
        \Gamma_k(\delta)\,
        \mathbb E_{t\sim[-1,1]}\!\left[|\widehat g(t)-g_\star(t)|^2\right].
$$
If $g_\star$ and $\widehat g$ are each $k$-sparse but may use different frequency sets, the same conclusion
holds with $\Gamma_{2k}(\delta)$.
\end{corollary}

\begin{proof}
The difference $h=\widehat g-g_\star$ belongs to $E_\Lambda\subseteq\Tcal_k$; apply the definition of
$\Gamma_k(\delta)$.  With unrelated supports, $h\in\Tcal_{2k}$.
\end{proof}

Consequently, any active or weighted least-squares algorithm that produces an in-domain error bound
$$
        \mathbb E_{t\sim[-1,1]}[|\widehat g(t)-g_\star(t)|^2]
        \leq\epsilon_{\mathrm{train}}^2
$$
automatically gives exterior integrated error at most
$\Gamma_k(\delta)\epsilon_{\mathrm{train}}^2$, and the upper bound in
Proposition~\ref{prop:exterior-leverage} makes this guarantee explicit.  For example, if a regression routine has the
standard parametric risk $O(\sigma^2k/m)$ after $m$ noisy labels, then the same estimator has exterior risk at most
$$
        O\!\left(\frac{\sigma^2k}{m}\,\bigl(1+k\sqrt\delta\bigr)e^{Ck\sqrt\delta}\right).
$$
The overhead is polynomial in $k$ when $k\sqrt\delta=O(\log k)$, and it is constant at the natural edge width
$\delta=O(k^{-2})$.  The Chebyshev witness provides a single signal throughout the exterior interval, so the lower bound applies to $\Gamma_k(\delta)$ itself. It follows that any
black-box implication based only on an in-domain $L^2$-error guarantee must incur super-polynomial amplification
when $k\sqrt\delta/\log k\to\infty$ (in the regime $\delta\geq k^{-2}$).

\section*{Acknowledgments}
We thank Xue Chen for helpful discussions and the anonymous SODA reviewers for valuable feedback.

\paragraph{AI Disclosure.} We first asked ChatGPT 5.5 Pro to improve the conjecture in~\cite{ChenPrice2019} with a guided prompt. Over multiple rounds of conversations, it explored several unsuccessful approaches, including the Christoffel function, Gram-matrix inversion, and Paley--Wiener interpolation. Then we sent these preliminary results to Claude Opus 4.8, which resolved the easy case using a Turán-type argument (Section~\ref{sec:outer}). We asked ChatGPT 5.5 Pro again, and the proof for $1<x<2$ using the weighted half-line/Laguerre route emerged. We reconstructed and verified every argument and every cited result, corrected errors in intermediate drafts with the assistance of ChatGPT 6 Pro, Claude Opus 5.5, and Claude Fable 5.1, and wrote the final proofs and exposition.  During the final revision, again with the help of ChatGPT 6 Pro, we found a simpler proof by frequency elimination, which is now the main argument (Section~\ref{sec:elimination}). The author takes full responsibility for the paper's correctness and presentation.

\appendix
\section{Connection to the Takenaka--Malmquist basis}
\label{app:analytic}

Section~\ref{sec:elimination} bounds $|g(1+r)|$ by the weighted energy ${\cal E}_\alpha(g)$ without ever constructing a basis.  This appendix makes the underlying structure explicit.  Fix distinct frequencies $\lambda_1,\ldots,\lambda_{k}$, and equip the span of $e^{i\lambda_1t},\ldots,e^{i\lambda_{k}t}$ with the inner product
\[
    \langle f_1,f_2\rangle_\alpha:=\int_{-\infty}^1f_1(t)\,\overline{f_2(t)}\,w_\alpha(t)\,dt,\qquad w_\alpha(t)=e^{\alpha(t-1)},
\]
so that $\langle f,f\rangle_\alpha={\cal E}_\alpha(f)$.\footnote{Throughout this appendix, we assume that $f:\R\rightarrow\C$ is continuous and of polynomial growth, i.e., $|f(t)|\leq C(1+|t|)^N$ for some $C,N$.}  From the reconstruction identity~\eqref{eq:reconstruction-identity} we build an explicit orthonormal basis of this span, the Takenaka--Malmquist basis.  We then show that the endpoint values $b_m(1)$ recorded by frequency elimination, divided by $\sqrt\alpha$, are exactly the coordinates of $g$ in this basis.  Finally, we compute the largest possible value of $|f(1+r)|^2/{\cal E}_\alpha(f)$ over all $f\in\Tcal_k$.  It equals $\alpha\sum_{m<k}L_m(-\alpha r)^2$, where $L_m$ is the Laguerre polynomial (Definition~\ref{def:def_laguerre}).  This value is approached as all frequencies approach a common value $\lambda$.  

\begin{theorem}[Laguerre half-line comparison]\label{thm:laguerre-halfline}
Let $f\in\Tcal_k$, let $\alpha>0$, and let $r\geq0$. Then
\begin{equation}\label{eq:halfline}
        |f(1+r)|^2\leq\alpha\sum_{m=0}^{k-1}L_m(-\alpha r)^2\,\mathcal E_\alpha(f).
\end{equation}
\end{theorem}

Lemma~\ref{lem:elimination}, Lemma~\ref{lem:reconstruction}, and Cauchy--Schwarz already imply \eqref{eq:halfline}. Lemma~\ref{lem:collision} shows that this bound is sharp for every $k$, $\alpha$, and $r$.

For $\lambda\in\R$ and continuous $h:\R\to\C$, define an operator
\begin{equation}\label{eq:Rmu}
        ({\cal J}_\lambda h)(t):=\int_1^t e^{i\lambda(t-s)}h(s)\,ds,
        \qquad \mathcal R_\lambda:=\operatorname{Id}+\alpha {\cal J}_\lambda.
\end{equation}
In this notation, \eqref{eq:reconstruction-identity} becomes
\[
        f_j(t)=b_j(1)e^{i\lambda_{j+1}(t-1)}+(\mathcal R_{\lambda_{j+1}}f_{j+1})(t).
\]
Thus $\mathcal R_{\lambda_{j+1}}$ reconstructs the part of $f_j$ determined by the next residual. 

The following lemma shows that ${\cal R}_\lambda$ is an isometry into the orthogonal complement of $e^{i\lambda t}$.

\begin{lemma}[Weighted isometry and orthogonality]\label{lem:Rmu}
Let $\lambda\in\R$, and let $h,f:\R\to\C$. Then
\[
        \langle\mathcal R_\lambda h,\mathcal R_\lambda f\rangle_\alpha=\langle h,f\rangle_\alpha,
        \qquad
        \langle\mathcal R_\lambda h,e^{i\lambda t}\rangle_\alpha=0.
\]
\end{lemma}

\begin{proof}
Set $u:={\cal J}_\lambda h$ and $v:={\cal J}_\lambda f$. Then
\[
        u'=h+i\lambda u,\qquad v'=f+i\lambda v,\qquad u(1)=v(1)=0.
\]
Using $\mathcal R_\lambda h=h+\alpha u$, $\mathcal R_\lambda f=f+\alpha v$, and $w_\alpha'=\alpha w_\alpha$, we obtain
\begin{align*}
    \langle {\cal R}_\lambda h, {\cal R}_\lambda f\rangle_\alpha - \langle h, f\rangle_\alpha= &~ \int_{-\infty}^1 w_\alpha(t) \left((\mathcal R_\lambda h)(t)\overline{(\mathcal R_\lambda f)(t)}-h(t)\overline {f(t)}\right)\,dt\\
    = &~ \int_{-\infty}^1 w_\alpha(t) \left(\alpha u(t) \overline{f(t)} + \alpha h(t) \overline{v(t)} + \alpha^2 u(t)\overline{v(t)}\right)\,dt\\
    = &~ \int_{-\infty}^1 \alpha (w_\alpha(t) u(t) \overline{v(t)})'\,dt=0.
\end{align*}
We also have
\begin{align*}
    \langle \mathcal R_\lambda h, e^{i\lambda t}\rangle_\alpha = &~ \int_{-\infty}^1 w_\alpha(t)e^{-i\lambda t}({\cal R}_\lambda h)(t)\,dt\\
    = &~\int_{-\infty}^1 w_\alpha(t)e^{-i\lambda t}\left(h(t) + \alpha u(t)\right)\,dt\\
    = &~ \int_{-\infty}^1 \left(w_\alpha (t) u(t)e^{-i\lambda t}\right)'\,dt = 0.
\end{align*}
\end{proof}

To determine the span generated by reconstruction, apply $\mathcal R_\lambda$ to a single exponential. For $\mu\neq\lambda$, direct integration gives
\begin{equation}\label{eq:Rmu-exp}
        \mathcal R_\lambda\bigl(e^{i\mu(t-1)}\bigr)
        =e^{i\mu(t-1)}+\alpha\frac{e^{i\mu(t-1)}-e^{i\lambda(t-1)}}{i(\mu-\lambda)}.
\end{equation}
Thus each operation adds at most its own frequency. The following lemma constructs the Takenaka--Malmquist basis using ${\cal R}_\lambda$.

\begin{lemma}[Takenaka--Malmquist basis {\cite{Takenaka1925,Malmquist1926}; cf.\ \cite{Nikolski2002,GarciaRoss2015}}]\label{lem:TM}
Let $\lambda_1,\ldots,\lambda_{k}$ be distinct real numbers, and define
\begin{equation}\label{eq:TM}
        q_m:=\mathcal R_{\lambda_1}\cdots\mathcal R_{\lambda_{m-1}}
        \bigl(\sqrt\alpha\,e^{i\lambda_m(t-1)}\bigr),
        \qquad 1\leq m\leq k.
\end{equation}
Then $q_m\in\operatorname{span}\{e^{i\lambda_1t},\ldots,e^{i\lambda_mt}\}$, with a nonzero coefficient on $e^{i\lambda_mt}$. The functions $q_1,\ldots,q_{k}$ form an orthonormal basis of $\operatorname{span}\{e^{i\lambda_jt}:1\leq j\leq k\}$ for $\langle\cdot,\cdot\rangle_\alpha$.
\end{lemma}

\begin{proof}
Applying \eqref{eq:Rmu-exp} successively, from the innermost operator outwards, shows that $q_m$ lies in the stated span. The initial exponential has norm one, since
\[
        \left\|\sqrt\alpha\,e^{i\lambda_m(t-1)}\right \|_\alpha^2
        =\alpha\int_{-\infty}^1w_\alpha(t)\,dt=1.
\]
By Lemma~\ref{lem:Rmu}, each ${\cal R}_{\lambda_j}$ is an isometry. Thus, $\|q_m\|_\alpha=1$.

Then, we show that $q_m$'s are orthogonal to each other. For any $\ell<m$, let
\[
        h:=\mathcal R_{\lambda_{\ell+1}}\cdots\mathcal R_{\lambda_{m-1}}
        \left(\sqrt\alpha\,e^{i\lambda_m(t-1)}\right).
\]
Using the isometry property, we have
\[
        \langle q_m,q_\ell\rangle_\alpha
        =\left\langle\mathcal R_{\lambda_\ell}h,
        \sqrt\alpha\,e^{i\lambda_\ell(t-1)}\right\rangle_\alpha=0
\]
by the orthogonality property in Lemma~\ref{lem:Rmu}. Hence, $q_1,\ldots,q_k$ are $k$ orthonormal functions in the $k$-dimensional span of $e^{i\lambda_1t},\dots,e^{i\lambda_kt}$.  In particular, $q_m$ has a nonzero coefficient on $e^{i\lambda_mt}$, since it is orthogonal to the span of the preceding exponentials.
\end{proof}

The following lemma bounds the exterior values of this basis $q_m$ in terms of the Laguerre polynomials, as in Section~\ref{sec:elimination}.

\begin{lemma}[Extrapolation bound for the basis]\label{lem:lag-majorant}
In the setting of Lemma~\ref{lem:TM}, for every $r\geq0$ and $1\leq m\leq k$,
\begin{equation}\label{eq:majorant}
        |q_m(1+r)|\leq\sqrt\alpha\,L_{m-1}(-\alpha r).
\end{equation}
\end{lemma}

\begin{proof}
The integral kernel in \eqref{eq:Rmu} has modulus one. Thus,
\begin{equation}\label{eq:Rmu-positive}
        |(\mathcal R_\lambda h)(1+r)|
        \leq|h(1+r)|+\alpha\int_0^r|h(1+s)|\,ds.
\end{equation}
Let $\Phi_n$ be defined as in Lemma~\ref{lem:reconstruction}. If $|h(1+s)|\leq\sqrt\alpha\,\Phi_n(s)$ for every $s\geq0$, then
\[
        |(\mathcal R_\lambda h)(1+r)|
        \leq\sqrt\alpha\left(\Phi_n(r)+\alpha\int_0^r\Phi_n(s)\,ds\right)
        =\sqrt\alpha\,\Phi_{n+1}(r).
\]
Hence, by~\eqref{eq:TM},
\begin{align*}
    |q_m(1+r)|=\left|\mathcal R_{\lambda_1}\cdots\mathcal R_{\lambda_{m-1}}\bigl(\sqrt\alpha\,e^{i\lambda_m(t-1)}\bigr)(1+r)\right|\leq \sqrt{\alpha}\Phi_{m-1}(r)=\sqrt{\alpha}L_{m-1}(-\alpha r). 
\end{align*}
\end{proof}

\begin{proof}[Proof of Theorem~\ref{thm:laguerre-halfline}]
Wlog, we may assume that $f$ has $k$ distinct frequencies. The case $f=0$ is immediate. Let $q_1,\ldots,q_{k}$ be the basis from Lemma~\ref{lem:TM}. Distinct exponentials are linearly independent on every interval, so the expansion
\[
        f=\sum_{m=1}^{k}\langle f,q_m\rangle_\alpha\,q_m
\]
holds as an identity of functions on $\R$. By Cauchy--Schwarz and the orthonormality of $\{q_m\}$, we have
\begin{align*}
        |f(1+r)|^2
        &\leq\left(\sum_{m=1}^{k}|q_m(1+r)|^2\right)
        \left(\sum_{m=1}^{k}|\langle f,q_m\rangle_\alpha|^2\right)\\
        &=\left(\sum_{m=1}^{k}|q_m(1+r)|^2\right)\mathcal E_\alpha(f)\\
        &\leq\alpha\sum_{m=0}^{k-1}L_m(-\alpha r)^2\,\mathcal E_\alpha(f),
\end{align*}
where the last inequality follows from Lemma~\ref{lem:lag-majorant}. The theorem then follows.
\end{proof}

The following lemma gives an explicit expansion of a Fourier signal in the Takenaka-Malmquist basis.
\begin{lemma}
Let $g\in {\cal T}_k$ with $k$ distinct frequencies $\lambda_1,\dots,\lambda_k$. Let $b_0(1),\dots,b_{k-1}(1)$  be defined as in Lemma~\ref{lem:elimination}. Let $q_1,\dots,q_k$ be the Takenaka--Malmquist basis for the span of $e^{i\lambda_1 t},\dots,e^{i\lambda_kt}$ defined in Lemma~\ref{lem:TM}. Then, $g$ can be expanded as
\begin{align*}
    g=\sum_{m=1}^k \frac{b_{m-1}(1)}{\sqrt{\alpha}} q_m,\qquad \langle g,q_m\rangle_\alpha = \frac{b_{m-1}(1)}{\sqrt{\alpha}}.
\end{align*}
\end{lemma}
\begin{proof}
Recall the frequency elimination process in Lemma~\ref{lem:elimination}:
\begin{align*}
    f_j=b_j(1)e^{i\lambda_{j+1}(t-1)}+\mathcal R_{\lambda_{j+1}}f_{j+1},\qquad f_0=g,\qquad f_k=0.
\end{align*}
Thus, we have
\begin{align*}
    g=\sum_{m=1}^{k}b_{m-1}(1)\,
        \mathcal R_{\lambda_1}\cdots\mathcal R_{\lambda_{m-1}}
        \bigl(e^{i\lambda_m(t-1)}\bigr)
        =\sum_{m=1}^{k}\frac{b_{m-1}(1)}{\sqrt\alpha}\,q_m.
\end{align*}
\end{proof}

The weighted comparison in Theorem~\ref{thm:laguerre-halfline} is sharp, and its upper bound is approached when the frequencies coalesce, as shown in the following lemma. 

\begin{lemma}[Collision limit and sharpness]\label{lem:collision}
For every $k\geq1$, $\alpha>0$, and $r\geq0$,
\begin{equation}\label{eq:exact-weighted}
        \sup_{0\neq f\in\Tcal_k}
        \frac{|f(1+r)|^2}{\mathcal E_\alpha(f)}
        =\alpha\sum_{m=0}^{k-1}L_m(-\alpha r)^2.
\end{equation}
The supremum is approached by signals with frequencies tending to any prescribed $\mu\in\R$.
\end{lemma}

\begin{proof}
The upper bound follows from Theorem~\ref{thm:laguerre-halfline}. 

For the lower bound, choose distinct frequencies $\lambda_j=\mu+(j-1)\varepsilon$, $1\leq j\leq k$, and let $q_m$ be the corresponding Takenaka--Malmquist basis from \eqref{eq:TM}. The signal
\[
        f_\varepsilon(t):=
        \sum_{m=1}^{k}
        \overline{q_m(1+r)}\,q_m(t)
\]
belongs to $\Tcal_k$. By orthonormality, its value at $1+r$ and its squared weighted norm both equal $\sum_m|q_m(1+r)|^2$. Hence
\[
        \frac{|f_\varepsilon(1+r)|^2}{\mathcal E_\alpha(f_\varepsilon)}
        =\sum_{m=1}^{k}|q_m(1+r)|^2.
\]

To evaluate the limit, observe that at a common frequency $\mu$, the reconstruction operator defined in~\eqref{eq:Rmu} acts on a modulated polynomial as
\begin{equation}\label{eq:Rmu-conj}
        \mathcal R_\mu\bigl(p(t)e^{i\mu(t-1)}\bigr)
        =e^{i\mu(t-1)}
        \left(p(t)+\alpha\int_1^t p(s)\,ds\right).
\end{equation}
Thus, when $\varepsilon\rightarrow 0$, by~\eqref{eq:TM} and the polynomial identity~\eqref{eq:Phi-formula}, we have
\begin{align*}
    q_m(t)\longrightarrow  \underbrace{\mathcal R_{\mu}\cdots\mathcal {R}_{\mu}}_{m-1~\text{times}}
        \bigl(\sqrt\alpha\,e^{i\mu(t-1)}\bigr)=\sqrt{\alpha}\Phi_{m-1}(t-1)e^{i\mu(t-1)}=\sqrt{\alpha} L_{m-1}(-\alpha(t-1))e^{i\mu(t-1)}.
\end{align*}
The convergence is uniform on compact intervals, since the operator formula involves finitely many iterated integrals whose kernels depend continuously on the frequencies. Evaluating at $t=1+r$ therefore yields
\[
        \lim_{\varepsilon\to0}
        \frac{|f_\varepsilon(1+r)|^2}{\mathcal E_\alpha(f_\varepsilon)}
        =\alpha\sum_{m=0}^{k-1}L_m(-\alpha r)^2,
\]
proving the matching lower bound.
\end{proof}

Up to a constant factor, the kernels in the proof converge on compact intervals to
\[
        e^{i\mu(t-1)}p_\star\bigl(\alpha(1-t)\bigr),
        \qquad
        p_\star(y):=\sum_{m=0}^{k-1}L_m(-\alpha r)L_m(y).
\]

\begin{remark}[The Hardy-space representation]\label{rem:hardy}
On the full weighted space $L^2\left(({-\infty},1],w_\alpha(t)\,dt\right)$, the map $h\mapsto H$ with $H(u):=e^{-\alpha u/2}h(1-u)$ is an isometry onto $L^2(0,\infty)$. In particular, $\int_0^\infty |H(u)|^2\,du={\cal E}_\alpha(h)$. Under this map, the reconstruction operator ${\cal R}_\lambda$ in~\eqref{eq:Rmu} becomes
\begin{align*}
    H(u)\longmapsto &~ H(u)-\alpha \int_0^{u} e^{-\zeta(u-v)} H(v)\,dv,\qquad \zeta := i\lambda+\alpha/2.
\end{align*}
By the Laplace transform $L[H](s)=\int_0^\infty H(u)e^{-su}\,du$, which turns convolution into multiplication, the image of ${\cal R}_\lambda$ becomes
\begin{align*}
    L[H]\longmapsto \frac{s-\overline\zeta}{s+\zeta}L[H].
\end{align*}

Consider the orthonormal basis from~\eqref{eq:TM}. The starting function $\sqrt{\alpha}e^{i\lambda_m(t-1)}$ is mapped to $\sqrt{\alpha}e^{-\zeta_mu}$ with $\zeta_m := i\lambda_m+\alpha/2$, whose Laplace transform is $\sqrt{\alpha}/(s+\zeta_m)$. Thus, the Laplace transform of the image of $q_m$ is
\[
        \frac{\sqrt\alpha}{s+\zeta_m}
        \prod_{\ell=1}^{m-1}\frac{s-\overline{\zeta_\ell}}{s+\zeta_\ell},
        \qquad \zeta_j=\frac\alpha2+i\lambda_j.
\]
These are the classical Takenaka--Malmquist functions in the Hardy space of the right half-plane \cite{Takenaka1925,Malmquist1926,Duren1970,Nikolski2002,GarciaRoss2015}. 
Both properties of $\mathcal R_\lambda$ in Lemma~\ref{lem:Rmu} become one-line facts.  First, $\left|\frac{i\omega -\overline{\zeta}}{i\omega+\zeta}\right|=1$ for real $\omega$, so by Plancherel multiplication by $\frac{s-\overline{\zeta}}{s+\zeta}$ preserves norms.  Second, $\langle H,e^{-\zeta u}\rangle=L[H](\overline\zeta)$, and this vanishes after multiplication by $\frac{s-\overline{\zeta}}{s+\zeta}$ because $\frac{\overline{\zeta}-\overline{\zeta}}{\overline{\zeta}+\zeta}=0$.
\end{remark}

\section{Plugging the optimal filter into the Chen--Price algorithm}
\label{app:cp-interface}

This appendix verifies that the filter $H_{k,T}$ from Theorem~\ref{thm:cheb-matched-filter} can replace the filter used in the proof of \cite[Theorem~1.5]{ChenPrice2019}.  Their argument needs the filter to satisfy the following conditions: 
\begin{enumerate}\renewcommand\labelenumi{(\arabic{enumi})}
    \item $H_{k,T}$ preserves a constant fraction of the in-window signal energy while suppressing exterior energy (Lemma~1.6 in \cite{ChenPrice2019});
    \item $|H_{k,T}(x)|\leq 1.01$ for any $x$ (item (3) of Theorem 4.2 in \cite{ChenPrice2019});
    \item $H_{k,T}$ can be efficiently evaluated at the queried points (implicitly used by line~5 of Algorithm~2 in \cite{ChenPrice2019}).  
\end{enumerate}
We will verify these conditions in this appendix: (1) will be verified in Lemma~\ref{lem:CP-global-energy}; (2) is already proved in Theorem~\ref{thm:cheb-matched-filter}; and  (3) will be verified in Lemma~\ref{lem:CP-window-evaluation}.
The only parameter that changes is the effective cluster half-width $\Delta_*:=\Delta+C_\eta\frac{k^2}{T}$.


\begin{lemma}[Energy preservation and concentration]
\label{lem:CP-global-energy}
Let $H_{k,T}$ be the filter from Theorem~\ref{thm:cheb-matched-filter}. Then, for every $g\in \mathcal{T}_k$,
\begin{equation}
\label{eq:CP-global-energy-conclusions}
\begin{aligned}
        \int_{-T}^{T}|H_{k,T}(t)g(t)|^2\,dt\geq &~ (1-\sqrt\eta)^2\int_{-T}^{T}|g(t)|^2\,dt\,,\\
        \int_{-T}^{T}|H_{k,T}(t)g(t)|^2\,dt\geq &~ \frac{(1-\sqrt\eta)^2}{(1-\sqrt\eta)^2+\eta}  \int_{\R}|H_{k,T}(t)g(t)|^2\,dt\,.
\end{aligned}
\end{equation}
In particular, for sufficiently small fixed $\eta$ (e.g., $\eta\leq1/400$), these imply the two
constant-factor energy properties in \cite[Lemma~1.6]{ChenPrice2019}.
\end{lemma}

\begin{proof}
Theorem~\ref{thm:cheb-matched-filter} gives
$$
    \underbrace{\int_{-T}^{T}|(1-H_{k,T}(t))g(t)|^2\,dt}_{E_{\rm defect}}
    +\underbrace{\int_{\R\setminus[-T,T]}|H_{k,T}(t)g(t)|^2\,dt}_{E_{\rm out}}
    \leq\eta\underbrace{\int_{-T}^{T}|g(t)|^2\,dt}_{E_0}.
$$
Let
$$
        E_{\rm in}:=\int_{-T}^{T}|H_{k,T}(t)g(t)|^2\,dt.
$$
The triangle inequality in $L^2[-T,T]$ implies
$$
        E_{\rm in}^{1/2}
        \geq E_0^{1/2}-E_{\rm defect}^{1/2}
        \geq(1-\sqrt\eta)E_0^{1/2},
$$
and hence
$$
        E_{\rm in}\geq(1-\sqrt\eta)^2E_0,
        \qquad
        \frac{E_{\rm in}}{E_{\rm in}+E_{\rm out}}
        \geq\frac{(1-\sqrt\eta)^2}{(1-\sqrt\eta)^2+\eta}\,,
$$
which give~\eqref{eq:CP-global-energy-conclusions} in the lemma.
\end{proof}

\begin{lemma}[Efficient evaluation]
\label{lem:CP-window-evaluation}
For every fixed $\eta$, the filter $H_{k,T}$ can be evaluated at every queried point in the algorithm in \cite{ChenPrice2019} to the required constant relative accuracy in $\operatorname{poly}(k)$ arithmetic time per query.
\end{lemma}

\begin{proof}
By scaling, it is enough to take $T=1$.  In the notation of the proof of Theorem~\ref{thm:cheb-matched-filter},
$$
        H_{A,B}(t)
        =\frac{\displaystyle\int_{t-1}^{t+1}\Psi_{A,B}(u)\,du}
               {\displaystyle\int_{\mathbb R}\Psi_{A,B}(u)\,du} = \int_{t-1}^{t+1}K_{A,B}(u)\,du\,,
        \qquad B=\Theta_\eta(k^2)\,,
$$
where $\Psi_{A,B}$ is the product of $O(\log k)$ even powers of $\sinc$ functions in \eqref{eq:sinc-product-filter}.  The Chen--Price phase test queries the filter only within
$O(1/k^2)=O(1/B)$ of $[-1,1]$, because $|\beta|\leq\gamma_{\mathrm{CP}}/\Delta_*$ for a constant $\gamma_{\mathrm{CP}}$ of their phase test, and
$\Delta_*\geq c_\eta k^2$ after normalization.  Decreasing the fixed constant $\gamma_{\mathrm{CP}}$, these points satisfy
$|t|\leq1+a/B$ for a sufficiently small constant $a>0$.

On the query region, $H_{A,B}$ is uniformly bounded away from zero.
Indeed, $K_{A,B}$ is a non-negative even probability density. Thus, it suffices to consider
$0\leq t\leq1+a/B$.  If $0\leq t\leq1$, then
$[0,1]\subseteq[t-1,t+1]$, and therefore
$$
        H_{A,B}(t) \geq \int_0^1K_{A,B}(u)\,du =\int_0^\infty K_{A,B}(u)\,du - \int_1^\infty K_{A,B}(u)\,du = \frac{1}{2}-P_{A,B}(1)\,.
$$
If $t=1+r$ with $0\leq r\leq a/B$, then
$$
\begin{aligned}
        H_{A,B}(1+r)
        &=
        \int_r^{2+r}K_{A,B}(u)\,du\\
        &=
        \frac12
        -\int_0^rK_{A,B}(u)\,du
        -P_{A,B}(2+r)\\
        &\geq
        \frac12
        -\frac{a}{B}\|K_{A,B}\|_{L^\infty}
        -P_{A,B}(1).
\end{aligned}
$$
Using $0\leq\Psi_{A,B}\leq1$ and \eqref{eq:Z-lower-filter}, we have
$$
        \|K_{A,B}\|_{L^\infty} \leq Z_{A,B}^{-1} \leq C_AB\,.
$$
Consequently, for any $|t|\leq 1+\frac{a}{B}$,
$$
        H_{A,B}(t)
        \geq
        \frac12-C_Aa-P_{A,B}(1)\,.
$$
Choose $a>0$ so that $C_Aa\leq1/8$, and then choose the constant in $B=\Theta_\eta(k^2)$ sufficiently large such that by~\eqref{eq:kernel-far-tail},
$$
        P_{A,B}(1) \leq C_Ae^{-cA\sqrt B} \leq \frac{1}{8}\,.
$$
Hence,
$$
        H_{A,B}(t)\geq\frac{1}{4}
$$
throughout the query region.

Furthermore,
$$
        \left|\Psi_{A,B}'(u)\right|\leq \frac{1}{2}\sum_{j=0}^{J}m_j\omega_j\leq C_AB
$$
provides an explicit Lipschitz constant for the integrand.  Hence, on an interval of length $L$, a composite rectangle rule with step size $h$ has an error of at most $C_ABLh$.  For $Z_{A,B}$, the tail estimates~\eqref{eq:kernel-far-tail} and \eqref{eq:kernel-root-tail} first reduce the integral over $\R$ to an integral over a bounded interval. Choosing $h=C_A\rho/B^2$ then approximates both the numerator and the denominator to
additive error $C_A\rho/B$, using $\operatorname{poly}(B,1/\rho)$
arithmetic operations. By~\eqref{eq:Z-lower-filter}, these additive-error approximations give the desired $O(\rho)$ relative-error approximation for $H_{A,B}(t)$.

Finally, we note that a sufficiently small constant $\rho$ relative error changes all importance-weighted energies by a $1+O(\rho)$ factor and
perturbs the phase-difference signal by $O(\rho)\|y_H\|_2$.  These errors are absorbed by the constant slack in
\cite[Lemma~8.2]{ChenPrice2019}.  Rescaling proves the lemma for general $T$.
\end{proof}

\bibliographystyle{alpha}
\bibliography{references}

\end{document}